\def\year{2019}\relax
\documentclass[letterpaper]{article}
\usepackage{aaai19}  
\usepackage{times}  
\usepackage{helvet}  
\usepackage{courier}  
\usepackage{url}  
\usepackage{graphicx}  

\usepackage{times}  
\usepackage{helvet}  
\usepackage{courier}  
\usepackage{url}  
\usepackage{lineno}
\usepackage{subfigure}
\usepackage{wrapfig,subfigure,tabularx,fancyvrb}
\usepackage{graphicx}
\usepackage[table]{xcolor}
\usepackage{multirow,booktabs}
\usepackage{color}
\usepackage{epstopdf}
\usepackage{bm}
\usepackage{scrextend}
\usepackage[ruled,vlined, linesnumbered, boxed]{algorithm2e}
\usepackage{amssymb}
\usepackage{amsmath}
\usepackage{amsthm}
\usepackage{lipsum}

\allowdisplaybreaks

\newtheorem{Lemma}{Lemma}

\newtheorem{Theorem}[Lemma]{Theorem}
\newtheorem{Definition}[Lemma]{Definition}

\newtheorem{Proposition}[Lemma]{Proposition}

\def\figwidth{0.48\columnwidth}

\begin{document}

\title{On the Inducibility of Stackelberg Equilibrium for Security Games}

\author{Qingyu Guo\textsuperscript{1}, Jiarui Gan\textsuperscript{2}, Fei Fang\textsuperscript{3}, Long Tran-Thanh\textsuperscript{4}, Milind Tambe\textsuperscript{5}, Bo An\textsuperscript{1}\\
\textsuperscript{1}School of Computer Science and Engineering, Nanyang Technological University, \{qguo005, boan\}@ntu.edu.sg\\
\textsuperscript{2}Department of Computer Science, University of Oxford, Jiarui.gan@cs.ox.ac.uk\\
\textsuperscript{3}School of Computer Science, Carnegie Mellon University, feifang@cmu.edu\\
\textsuperscript{4}Department of Electronics and Computer Science, University of Southampton, ltt08r@ecs.soton.ac.uk\\
\textsuperscript{5}Center for Artificial Intelligence in Society, University of Southern California, tambe@usc.edu
}
\maketitle

\begin{abstract}
\emph{Strong Stackelberg equilibrium} (SSE) is the standard solution concept of Stackelberg security games.
As opposed to the \emph{weak Stackelberg equilibrium} (WSE), the SSE assumes that the follower breaks ties in favor of the leader and this is widely acknowledged and justified by the assertion that the defender can \emph{often} induce the attacker to choose a preferred action by making an infinitesimal adjustment to her strategy. Unfortunately, in security games with resource assignment constraints, the assertion might not be valid; it is possible that the defender cannot induce the desired outcome.
As a result, many results claimed in the literature may be overly optimistic.
To remedy, we first formally define the utility guarantee of a defender strategy and provide examples to show that the utility of SSE can be higher than its utility guarantee. Second, inspired by the analysis of leader's payoff by~\citeauthor{von04}~\shortcite{von04}, we provide the solution concept called the \emph{inducible Stackelberg equilibrium} (ISE), which owns the highest utility guarantee and always exists. Third, we show the conditions when ISE coincides with SSE and the fact that in general case, SSE can be extremely worse with respect to utility guarantee.
Moreover, introducing the ISE does not invalidate existing algorithmic results as the problem of computing an ISE polynomially reduces to that of computing an SSE.
We also provide an algorithmic implementation for computing ISE, with which our experiments unveil the empirical advantage of the ISE over the SSE.
\end{abstract}

\section{Introduction}
\label{sec:intro}
The past few years have witnessed the huge success of game theoretic reasoning in the security domain~\cite{tambe11,An17}. Models based on the Stackelberg security game (SSG) have been deployed to protect high-profile infrastructures, natural resources, large public events, etc. (e.g.,~\cite{tsai09,Jain10,Yin14,Fang16,Basilico17}).
An SSG models the interaction between a defender and an attacker, where the defender commits to a mixed strategy first, and the attacker best responds with knowledge of the defender strategy.

The Stackelberg equilibrium is the standard solution concept for Stackelberg games~\cite{leitmann78}. In such an equilibrium, no player has the incentive to deviate and the leader assumes that deviations made will result in optimal responses of the follower when evaluating the benefit of deviations.
The tie-breaking rules differentiates two forms of Stackelberg equilibria.
The \emph{strong} form of the Stackelberg equilibrium, called the \emph{strong Stackelberg equilibrium} (SSE), assumes that the follower always breaks ties by choosing the best action for the defender, whereas its counterpart, the \emph{weak Stackelberg equilibrium} (WSE), assumes that the follower always chooses the worst action.
The SSE is commonly adopted as the standard solution concept because the WSE may not exist~\cite{von04}; and the counter-intuitive tie-breaking rule
is justified, implicitly or explicitly in the literature, by the assertion that the defender can often induce the favorable strong equilibrium by selecting a strategy arbitrarily close to the equilibrium.

Unfortunately, the assertion may break, especially in scenarios with various resource assignment constraints, such as scheduling constraints in the Federal Air Marshals Service (FAMS) domain, constraints on patrol paths for protecting ports, and constraints in the form of protection externalities~\cite{tsai09,Jain10,Shieh12,gan15}.
Most existing works failed to realize the potential impossibility to induce SSE in such domains.
If the desired SSE cannot be induced, results claimed would questionably be overly optimistic. Such overoptimism is problematic in its own right and may even cause greater risks for the following reasons.
First, these results may be used in making security resource acquisition decisions, i.e., what combination of security resources need to be procured~\cite{McCarthy16}; overoptimism of SSE may cause an insufficient number or wrong types of resources to be deployed.
Second, statements made based on comparisons between the expected utility of SSE with some heuristic strategies or human-generated solutions to claim superiority of SSE strategies would be in potential jeopardy~\cite{Pita08,tsai09,Xu17}.
Third, the SSE strategy recommended may not be the optimal one, thus failing in optimizing the use of limited security resources, which is the primary mission of security games.

In this paper, we remedy the inadequacy of the SSE in security games and make the following key contributions.
1) We formalize the notion of overoptimism by defining the utility guarantee of the defender's strategies, and show with a motivating example that the utility claimed to be guaranteed by the SSE is much higher than the actually guaranteed utility.
2) Inspired by the notion of inducible strategy~\cite{von04}, we characterize the solution concept with the highest utility guarantee and call it \emph{inducible Stackelberg equilibrium} (ISE).
3) We compare ISE with SSE and show that for games with certain structures, the two concepts are equivalent, though in general cases the guaranteed utility of SSE can be arbitrarily worse than that of ISE;
in addition, introducing the ISE does not invalidate existing algorithmic results as the problem of computing an ISE polynomially reduces to that of computing an SSE.
4) We provide algorithmic implementation for computing the ISE and conduct experiments to evaluate our results; our experiments unveil the significant overoptimism and sub-optimality  of the SSE, which suggests the practical significance of the ISE solution.

\subsubsection{Other Related Works}
To the best of our knowledge, \citeauthor{Okamoto12}~\shortcite{Okamoto12} are the only exception who have raised the concern of lack of inducibility in security games, though their model is a very specific type of network security games that cannot be generalized to standard security games, especially games with scheduling constraints.
Besides that, the more important question regarding the overoptimism due to the lack of inducibility and the algorithmic remedies needed for such overoptimism were left unanswered 
(in particular, the solution algorithm proposed by \citeauthor{Okamoto12} only converges to a local optimum even only in their setting). 
These questions are addressed in the affirmative in this paper.
%
The concept of \textit{inducible target} in our paper (Definition~\ref{def:inducible_target}) is inspired by \textit{inducible strategy} first proposed by von Stengel and Zamir~\shortcite{von04} in their study of general Stackelberg games. 
However, the focus of their work was solely on characterizing the range of leader's utility in Stackelberg equilibria with the aim of confirming the advantage of commitment~\cite{von04,Stengel10}. 
%
Some other works considered potential deviation of the attacker from their optimal responses and proposed solution concepts that were robust to these deviations~\cite{Pita09,Yang14,nguyen13}.
Our work differs from this line of research in that we consider perfectly rational attackers.

\section{Preliminaries}
\subsection{Security Games with Arbitrary Schedules}
A security game is a two-player Stackelberg game played between an attacker and a defender. The defender allocates resources $R$ to protect a set of targets $T$. Let $n=|T|$. A resource $r\in R$ can be assigned to a schedule $s\subseteq T$ which covers multiple targets and is chosen from a known and constrained set $S_r\subseteq 2^T$. The attacker's pure strategy is choosing one target $t\in T$ to attack, and his mixed strategy can be represented as a vector $\mathbf{a}\in\mathcal{A}$ where $a_t$ denotes the probability of attacking $t\in T$.
The defender's pure strategy is a joint schedule $j$ which assigns each resource to at most one schedule. Let $j$ be represented as a vector $\mathbf{P}_j=\langle P_{jt}\rangle\in\{0,1\}^n$ where $P_{jt}$ indicates whether target $t$ is covered in joint schedule $j$. The set of all feasible joint schedules is denoted by $J$. The defender's mixed strategy $\mathbf{x}\in\mathcal{X}$ is a vector where $x_j$ denotes the probability of playing joint schedule $j$. Let $\mathbf{c}=\langle c_t\rangle$ be the coverage vector corresponding to $\mathbf{x}$, where $c_t=\sum_{j\in J}P_{jt}x_j$ is the marginal probability of covering $t$.

The payoffs of players are decided by the target chosen by the attacker and whether the target is protected by the defender. The defender's payoff for an uncovered attack is denoted by $U^u_d(t)$ and for a covered attack $U^c_d(t)$. Similarly, $U^u_a(t)$ and $U^c_a(t)$ are attacker's payoffs respectively. A widely adopted assumption in security games is that $U^c_d(t)>U^u_d(t)$ and $U^u_a(t)>U^c_a(t)$. In other words, covering an attack is beneficial for the defender, while hurts the attacker. Given a strategy profile $\langle\mathbf{x},\mathbf{a}\rangle$, the expected utilities for both players are
\begin{equation*}
\begin{aligned}
U_d(\mathbf{x},\mathbf{a})&=\sum\nolimits_{t\in T}a_t[c_tU^c_d(t)+(1-c_t)U^u_d(t)]\\
U_a(\mathbf{x},\mathbf{a})&=\sum\nolimits_{t\in T}a_t[c_tU^c_a(t)+(1-c_t)U^u_a(t)],
\end{aligned}
\end{equation*}
where $\mathbf{c}$ is the coverage vector corresponding to $\mathbf{x}$. Let $U_a(\mathbf{x},t)$ and $U_d(\mathbf{x},t)$ denote the expected utilities of the attacker and defender respectively when $t$ is attacked.
The illustrated security game model has a wide applicability in many security applications~\cite{Kiekintveld09,Jain10,tsai09,gan15}.
\subsection{Stackelberg Equilibria and Tie-Breaking Rules}
\label{sec:relatedwork}
In an SSG, the defender acts first by committing to a mixed strategy and the attacker moves after having observed the defender's commitment. 
The solution concept of Stackelberg games, called \emph{Stackelberg equilibrium}, captures the outcome in which the defender's strategy is optimal, under the assumption that the attacker will always respond optimally to the strategy the defender plays~\cite{leitmann78}. 
A pair of strategies  $\langle\mathbf{x}^*, f(\mathbf{x}^*) \rangle$ forms a Stackelberg equilibrium iff:
\begin{enumerate}
  \item $f:\mathcal{X}\rightarrow \mathcal{A}$ is a best response function of the attacker, that satisfies: $U_a(\mathbf{x},f(\mathbf{x}))\geq U_a(\mathbf{x},\mathbf{a})$ for all $\mathbf{x}\in\mathcal{X}$ and $\mathbf{a}\in\mathcal{A}$;
      
  \item $U_d(\mathbf{x}^*,f(\mathbf{x}^*))\geq U_d(\mathbf{x},f(\mathbf{x}))$ for all $\mathbf{x}\in\mathcal{X}$.
\end{enumerate}
A tie represents a situation where multiple best response strategies exist for the attacker. Ties are not rare corner cases, but a fundamentally recurring situation in security games. To achieve maximal usage of defense resources, algorithms avoid allocating too many or too few resources to each target, and in most cases generate a tied solution~\cite{Paruchuri08,Kiekintveld09}. Thus, a tie-breaking rule -- how the attacker breaks ties -- plays a central role in security games and is exploited to design efficient algorithms, such as ORIGAMI~\cite{Kiekintveld09}. Different tie-breaking rules lead to different Stackelberg equilibria. The \emph{strong Stackelberg Equilibrium} (SSE) and the \emph{weak Stackelberg Equilibrium} (WSE) are two prevailing solution concepts, defined respectively with the optimistic and pessimistic assumptions of the attacker's tie-breaking behavior:
\begin{itemize}
  \item \textbf{SSE}: $f^{S}(\mathbf{x})\in\arg\max_{t\in\Gamma(\mathbf{x})}U_d(\mathbf{x},t)$ for every $\mathbf{x}$;

  \item \textbf{WSE}: $f^{W}(\mathbf{x})\in\arg\min_{t\in\Gamma(\mathbf{x})}U_d(\mathbf{x},t)$ for every $\mathbf{x}$;
\end{itemize}
where $\Gamma(\mathbf{x})=\arg\max_{t\in T}U_a(\mathbf{x},t)$ is the \emph{attack set}, the set of all best response pure strategies (targets) for attacker. In words, the attacker breaks ties in favor of the defender in the SSE, while against the defender in WSE.

\subsubsection{WSE and SSE}
WSE follows the spirit of maximin solution~\cite{sandholm15}, which provides the defender a guaranteed value in the sense that if the attacker breaks ties in a different manner, the defender does not gain less.
The SSE, however, does not provide such a value guarantee.
Despite this, the security game literature has adopted SSE instead of WSE primarily because a WSE may not exist \cite{Conitzer06}.
{In addition, the counter-intuitive assumption that the attacker breaks ties in favor of the defender is justified by the assertion that the desired outcome can \textit{often} be induced by playing a strategy arbitrarily close to the SSE strategy.}
%
Kiekintveld et al.~\shortcite{Kiekintveld09} were the {first} who explicitly made such a claim in the security game domain, following the analysis for generic Stackelberg games \cite{von04}.
%
%
{Since then, despite a lack of systematic research}, the claim has been commonly used to support the SSE in security games of various types, including games with {scheduling constraints} (e.g., \cite{Jain10,Varakantham13,gan15}).
{The idea of SSE is also integrated in real world systems such as the ARMOR deployed at LAX~\cite{Pita08}, and IRIS for the Federal Air Marshal Services~\cite{tsai09}.}
To see what can go wrong with the SSE assumption, we provide a concrete example in the next section.
\section{Motivating Example}
\label{sec:motivatingexample}
Consider an instance shown in the following figure where $T=\{t_1,t_2,t_3,t_4\}$. The defender has one resource $R=\{r\}$.
We first consider the scenario without resource assignment constraints, which has a unique SSE with coverage $\mathbf{c}=\langle\frac{4}{15},\frac{1}{5},\frac{4}{15},\frac{4}{15}\rangle$. In SSE, the attacker will break the tie $\Gamma(\mathbf{c})=T$ by attacking $t_2$.
This can be induced by decreasing the coverage on $t_2$ with infinitesimal amount and increasing the coverage on other targets, making $t_2$ be strictly preferred.
\begin{wrapfigure}{l}{35mm}
  \begin{center}
    \includegraphics[width = 40mm]{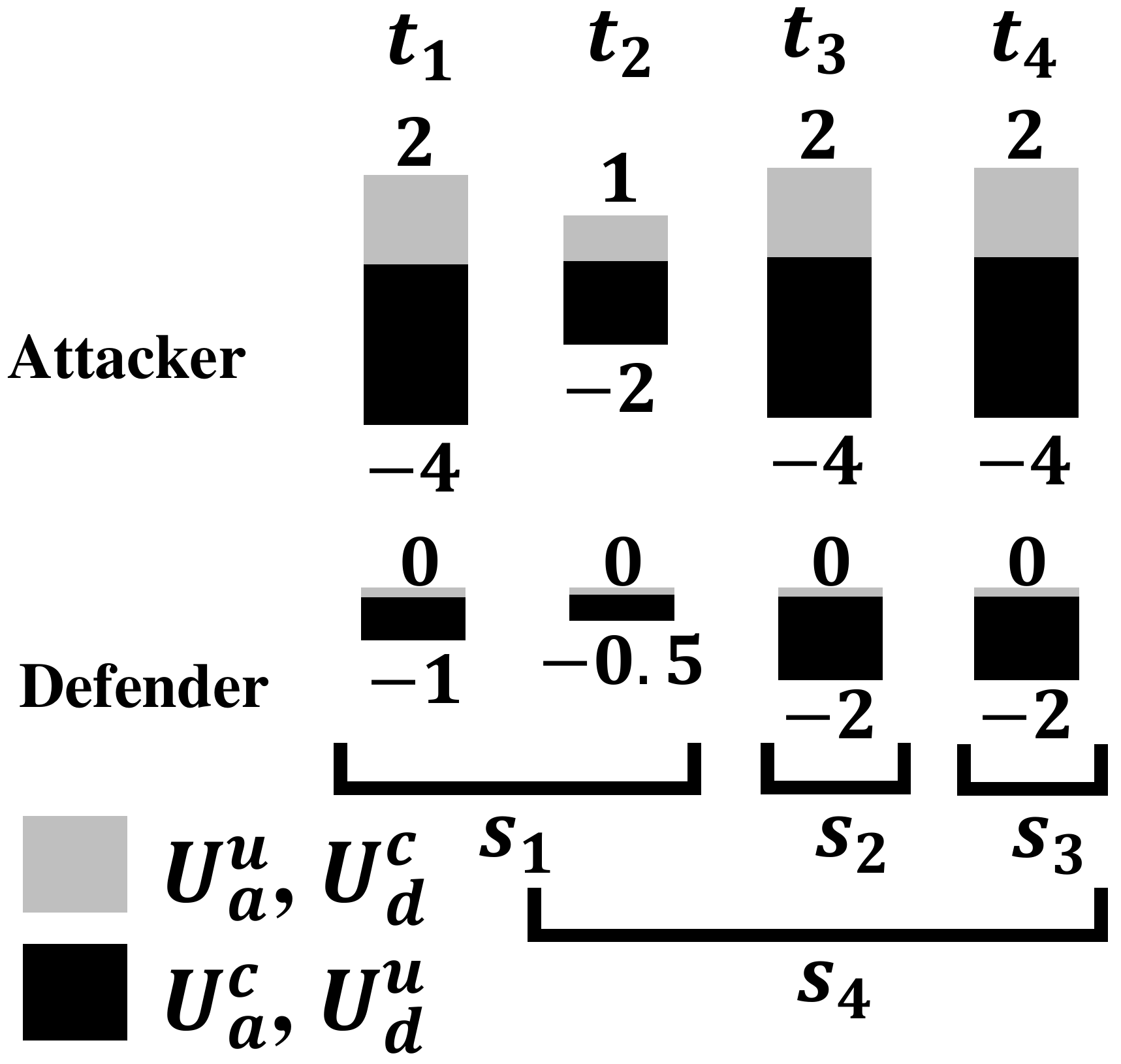}
  \end{center}
\end{wrapfigure}
However, with resource assignment constraints, the defender cannot decrease the coverage on one target arbitrarily while simultaneously not decreasing coverage on all other targets. Suppose joint schedules $J=\{s_1,s_2,s_3,s_4\}$ as shown in the figure. (There is only one resource.) The game still has a unique SSE where the defender plays $\mathbf{x}=\langle\frac{1}{3},\frac{1}{3},\frac{1}{3},0\rangle$ and the attacker is assumed to attack $t_2$, bringing the defender an expected utility of $-\frac{1}{3}$. Such outcome is explicitly or implicitly considered with previous mentioned infinitesimal strategy deviation in security game
literature~\cite{Jain10}. Unfortunately, there exists no strategy arbitrarily close to $\mathbf{x}$ which makes $t_2$ be strictly preferred by the attacker. If $x_1$ is decreased, the attacker will prefer $t_1$ over $t_2$; otherwise $t_3$ or $t_4$ will be attacked. Thus, any infinitesimal strategy deviation will cause the attacker to attack $t_1$, $t_3$ or $t_4$. The best induced outcome for the defender is only approaching $-\frac{2}{3}$, achieved by decreasing $x_1$ with infinitesimal amount and the attacker is induced to attack $t_1$.

Can the defender do better than $-\frac{2}{3}$? The answer is yes. Consider the mixed strategy $\langle\frac{1}{2},0,0,\frac{1}{2}\rangle$. The attack set is $\{t_1,t_3,t_4\}$ and the defender can induce the attacker to strictly prefer $t_1$ by playing $\langle\frac{1}{2}-\delta,0,0,\frac{1}{2}+\delta\rangle$ with infinitesimal $\delta$. By doing this, the defender can guarantee an expected utility arbitrarily close to $-0.5$, better than $-\frac{2}{3}$. In fact, this is the best outcome that the defender can achieve with infinitesimal strategy deviation. Such optimal outcome is captured by the solution concept called \emph{inducible Stackelberg equilibrium} (ISE), proposed in the following section.
\section{Inducible Stackelberg Equilibrium}

The above example reveals a failure of the attempt to induce the desired SSE outcome by playing a strategy arbitrarily close to the SSE strategy. It is natural to ask: Given any strategy $\mathbf{x}$, what is the best outcome inducible by playing strategies arbitrarily close to $\mathbf{x}$? Associated with such best outcome, which strategy is optimal? To answer these questions, inspired by the ``pessimistic" view of the leader's payoff in Stackelberg games~\cite{von04}, we define the \emph{utility guarantee} of a defender strategy as the supremum of the worst-case expected utility that can be achieved by playing a strategy arbitrarily close to the measured one.

\begin{Definition}[\textbf{Utility Guarantee}]\label{def:implementable}
The utility guarantee of a defender strategy $\mathbf{x}$ is defined as
\begin{equation}\label{eq:ug}
U^I(\mathbf{x})=\limsup_{\mathbf{x}'\rightarrow\mathbf{x}}\underset{t\in\Gamma(\mathbf{x}')}{\min}U_d(\mathbf{x}',t)
\end{equation}
\end{Definition}

The utility guarantee is well-defined since the limit superior always exists. It measures the inducibility of a defender strategy: $U^I(\mathbf{x})$ is the optimal outcome at $\mathbf{x}$ that is inducible via infinitesimal strategy deviation. The aforementioned assumption widely acknowledged in security games falsely claim that any SSE strategy $\mathbf{x}$ provides utility guarantee $U_d(\mathbf{x},f^S(\mathbf{x}))$, i.e., $U^I(\mathbf{x}) = U_d(\mathbf{x},f^S(\mathbf{x}))$. Therefore, we need to find the optimal strategy with respect to the utility guarantee. We notice that the optimal utility guarantee coincides with the ``pessimistic" leader's payoff~\cite{von04} as follows:
\begin{equation}
\max_{\mathbf{x}\in\mathcal{X}}U^I(\mathbf{x})=\underset{\mathbf{x}\in\mathcal{X}}{\sup}\min_{t\in\Gamma(\mathbf{x})}U_d(\mathbf{x},t)
\end{equation}
Inspired by the analysis of ``pessimistic" leader's payoff~\cite{von04}, we introduce several useful  notions for defining ISE. The first is \emph{inducible target}.
\begin{Definition}[\textbf{Inducible Target}]\label{def:inducible_target}
A target $t$ is inducible iff there exists at least one defender mixed strategy $\mathbf{x}\in\mathcal{X}$ such that $\Gamma(\mathbf{x}) = \{t\}$.
\end{Definition}
Inducible target offers the defender a lower bound on the utility guarantee as $U^I(\mathbf{x})\geq U_d(\mathbf{x},t)$ holds for any inducible target $t$ in $\Gamma(\mathbf{x})$. The intuition is as follows. Since $t$ is inducible, there exists $\mathbf{x}'$ against which $t$ is the unique best response for attacker. Thus, from $\mathbf{x}$, we can always play $(1-\alpha)\mathbf{x}+\alpha\mathbf{x}'$ with $\alpha\rightarrow0$ which always makes $t$ a unique best response as long as $\alpha>0$. Then it is easy to verify that the supremum in~\eqref{eq:ug} is always at least $U_d((1-\alpha)\mathbf{x}+\alpha\mathbf{x}',t)$, of which the limit is $U_d(\mathbf{x},t)$.

The concept of inducible target is insufficient to fully characterize the utility guarantee of a strategy because a pair of targets might be indistinguishable from the attacker's perspective as they always bring the attacker the same utility irrespective of the strategy the defender plays. Such targets are called \emph{identical targets}.
\begin{Definition}[\textbf{Identical Target}]\label{def:identical}
A pair of targets $t$ and $t'$ are identical iff $U_a(\mathbf{x},t)=U_a(\mathbf{x},t')$ for any $\mathbf{x}\in\mathcal{X}$.
\end{Definition}
Identical targets are non-inducible by Definition 2. However, it is possible that the optimal utility guarantee in~\eqref{eq:ug} is achieved via infinitesimal strategy deviation that induces a group of identical targets to be ``unique" best responses. Therefore, a more generalized notion, \emph{inducible element}, is defined to capture this special case. We begin with defining an \emph{element}.
\begin{Definition}[\textbf{Element}]\label{def:element}
An element is a set of targets in which: i) every pair of targets are identical, and ii) no target is identical to any target not in it.
\end{Definition}
The reason that we call it an element is as follows. First, from the attacker's perspective, the \emph{element} is the generalization on \emph{target} as it characterizes the extend to which the attacker can distinguish from the perspective of payoffs. Second, with mild assumption which often holds true in practice, one can easily verify that two targets are identical iff they have same payoffs for attacker and they are covered by the same set of schedules. Thus it is easy to enumerate all possible elements. Let $\{t\}$ be a singleton element if no target in $T$ is identical to $t$. The inducible element extends the concept of inducible targets as follows.
\begin{Definition}[\textbf{Inducible Element}]\label{def:inducible_element}
An element $e$ is inducible iff there exists at least one defender mixed strategy $\mathbf{x}\in\mathcal{X}$ such that $\Gamma(\mathbf{x}) = e$.
\end{Definition}
The observation that inducible target offers a lower bound to utility guarantee extends to inducible element. To show this, we first define the utility function in an element-based manner. For a defender strategy $\mathbf{x}$, we define $\tilde{U}_d(\mathbf{x},e)$ and $\tilde{U}_a(\mathbf{x},e)$ as follows
\begin{equation}\label{eq:element_utility}
\begin{aligned}
&\tilde{U}_d(\mathbf{x},e)=\min_{t\in e}U_d(\mathbf{x},t)\\
&\tilde{U}_a(\mathbf{x},e)=U_a(\mathbf{x},t)\quad\forall t\in e.
\end{aligned}
\end{equation}
One key observation here is that, if $e$ is inducible, $\tilde{U}_d(\mathbf{x},e)$ lower bounds $U^I(\mathbf{x})$. This follows the similar explanation with inducible targets. Since $e$ is inducible, there exists $\mathbf{x}'$ such that $\Gamma(\mathbf{x}')=e$ by definition and we can ``perturb" $\mathbf{x}$ towards $\mathbf{x}'$ with infinitesimal amount and the attack set becomes exactly $e$. The observation follows as we notice that $\tilde{U}_d(\mathbf{x},e)$ is a smooth function and thus the change on it with infinitesimal deviation on $\mathbf{x}$ is bounded.

With singleton element defined, the target set $T$ is partitioned into a disjoint element set $\mathcal{E}$. It is easy to see that, for any defender strategy $\mathbf{x}$, the attack set $\Gamma(\mathbf{x})$ is always a union of some elements in $\mathcal{E}$. Thus, we define $\tilde{\Gamma}(\mathbf{x})=\{e\in\mathcal{E}\mid e\subseteq\Gamma(\mathbf{x})\}$, and one can always verify that $\Gamma(\mathbf{x})=\bigcup_{e\in\tilde{\Gamma}(\mathbf{x})}e$. $\tilde{\Gamma}(\mathbf{x})$ can be interpreted as an ``attack set" consisting of elements, instead of targets. Let $\mathcal{E}^I\subseteq\mathcal{E}$ denote the set of inducible elements. The utility guarantee is actually decided by the inducible elements as presented in the following equation. 
\begin{equation}\label{eq:ugie}
U^I(\mathbf{x})=\max_{e\in\tilde{\Gamma}(\mathbf{x})\cap\mathcal{E}^I}\tilde{U}_d(\mathbf{x},e)
\end{equation}
The correctness of this equation formally follows the analysis of ``pessimistic" leader's payoff by~\citeauthor{von04}~\shortcite{von04}, and here we provide an intuitive explanation. Since the players' utility functions are smooth, with infinitesimal strategy deviation, $U_d(\mathbf{x},t)$ and $U_a(\mathbf{x},t)$ can be regarded as unchanged for any $t$. Besides, with infinitesimal strategy deviation, it is only possible for defender to ``remove" some target out of the attack set, while it is unable for the defender to add a new target into the attack set given the non-zero gap between attacker's utilities between targets inside and outside the attack set respectively. Thus, since the inducible outcome $U^I(\mathbf{x})$ is defined on the worst tie-breaking rule, i.e., $\underset{t\in\Gamma(\mathbf{x}')}{\min}U_d(\mathbf{x}',t)$ in~\eqref{eq:ug}, and infinitesimal strategy deviation won't change the defender's utility on any target, the defender always has an intention to reduce the attack set via infinitesimal strategy deviation. It is then noticed that, any inducible element can be the attack set itself with infinitesimal strategy deviation as we shown before, while any element, that is not inducible, cannot be the unique best response element. Besides, the definition of element determines that if one target from element $e$ is in the attack set, so as all targets from $e$. Thus, the defender can only get $\min_{t\in e}U_d(\mathbf{x},t)$ under worst-case tie-breaking rule, when $e\in\tilde{\Gamma}(\mathbf{x})\cap\mathcal{E}^I$ becomes the unique best response element with infinitesimal strategy deviation. Equation~\eqref{eq:ugie} then follows.

To this end, we successfully characterize the inducible outcome with well-defined concept of inducible elements, and we are ready to define the concept of inducible Stackelberg equilibrium, which straightforwardly follows the previous analysis
\begin{Definition}[\textbf{ISE}]\label{def:ise}
A pair of strategies $\langle\mathbf{x}^*,  f^{I}(\mathbf{x}^*) \rangle$ forms an ISE if the following holds:
\begin{enumerate}
\item $\mathbf{x}^*\in\arg\max\limits_{\mathbf{x}\in\mathcal{X}}U^I(\mathbf{x})$;
\item $f^I(\mathbf{x})\in\arg\min\limits_{t\in e(\mathbf{x})}U_d(\mathbf{x},t)$ where $e(\mathbf{x})\in\arg\max\limits_{e\in\tilde{\Gamma}(\mathbf{x})\cap\mathcal{E}^I}\tilde{U}_d(\mathbf{x},e)$.
\end{enumerate}
\end{Definition}
Tie-breaking rule $f^I$ partially shares the property of $f^S$ as the attacker breaks the ties of elements in favor of the defender. Meanwhile, it behaves as $f^W$ when the attacker breaks the ties of targets from the same element. Notice that ISE successfully addresses the inducibility issue of SSE, and always exists by its definition. In the next section, we conduct extensive analysis to compare ISE with SSE.
\section{ISE vs. SSE}
In this section, we formally show that when \emph{Subsets of Schedules Are Schedules} (SSAS)~\cite{Korzhyk11} is satisfied, ISE and SSE are equivalent under mild assumption. However, in general cases the utility guarantee of SSE can be much worse than that of ISE, and we present one such example.

Formally, SSAS states that $2^{s} \subseteq S$ for all $s\in S$.
This can happen, for example, when the defender can choose to bypass arbitrary targets on their patrol route.


\begin{Theorem}\label{theorem:SSAS}
If SSAS is satisfied, every SSE $\langle\mathbf{x}^*, t^*\rangle$ such that $c_{t^*}>0$ is also an ISE.
\end{Theorem}
\begin{proof}
It is easy to see that when SSAS is satisfied, no pair of targets are identical, as every target $t$ covered by $s$ is uniquely covered by $s'=\{t\}$. Therefore, $\mathcal{E}=\{\{t\}\mid\forall t\in T\}$. We then show that each singleton element $\{t\}$ is inducible. In other words, $t$ is an inducible target.
If $\Gamma (\mathbf{x}^*)$ contains only $t^*$, then $t^*$ is inducible.
If $\Gamma (\mathbf{x}^*)$ also contains other targets, since SSAS is satisfied, we can construct a defender strategy $\mathbf{x}'$ such that $x'_j = x_j^*$ for all supporting strategies $j$ that does not contain $t^*$, and $x'_{j\setminus\{s\}} = x_j^*$ for all other $j$ that contains $s$.
Thus, the corresponding coverage $c'_{t^*}$ strictly decreases while the coverage of other targets remain the same.
As a result, the attacker will strictly prefers to attack $t^*$, so $t^*$ is inducible.
%
\end{proof}

Under SSAS, the set of SSE strategies is also a subset of NE strategies~\cite{Korzhyk11}.
This suggests the relationship between SSE, ISE and NE strategies illustrated in Figure~\ref{fig:relationship:ISEvsSSEvsNE}.
Notice that security games without schedules can be seen as ones with singleton schedules $S=T$, so that SSAS is satisfied trivially.
Although SSAS is valid in many real scenarios, it is risky to regard it as being ubiquitous.
For example, in the presence of {protection externalities} \cite{gan15,gan17security}, the effect that a defense resource might protect a set of targets within a certain radius can hardly be confined to a specific subset; in FAMS tasks \cite{tsai09,Jain10}, when air marshals are allocated to a row of connected flights, it is unrealistic to make them ``jump'' over only a subset of the schedule.
Our example below shows that in general security games, SSE can be arbitrarily worse than ISE in terms of the utility guarantee.
\begin{figure}
  \centering
  \subfigure[SSE, ISE and NE]{
    \label{fig:relationship:ISEvsSSEvsNE} 
    \includegraphics[height=30mm]{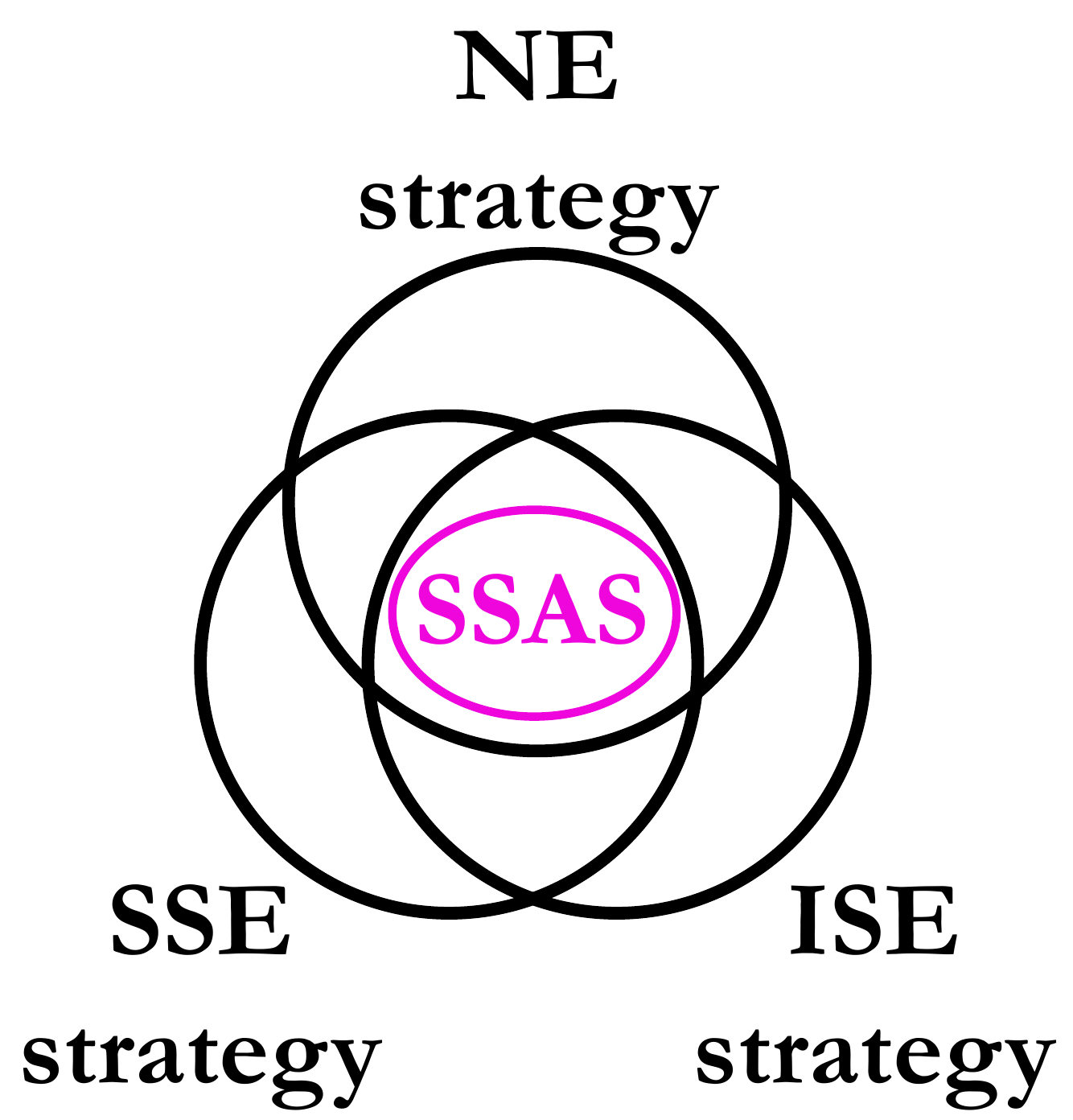}}
\subfigure[SSE can be arbitrarily worse]{
    \label{fig:relationship:ISEvsSSE} 
    \includegraphics[height=30mm]{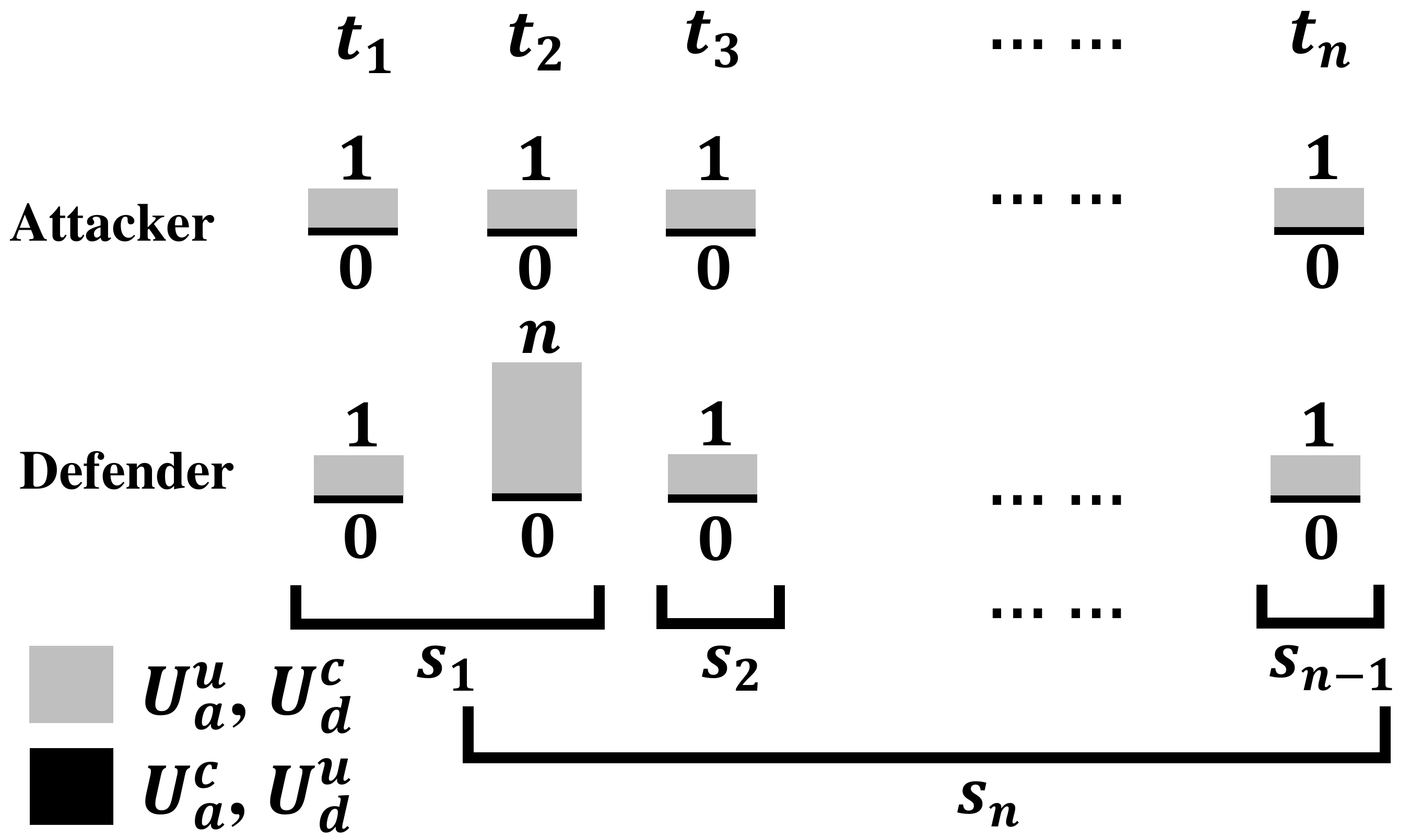}}
\caption{ISE vs. SSE.}\label{fig:relationship}
\end{figure}


\smallskip
\noindent
\textbf{Example1. }
The example is shown in Figure~\ref{fig:relationship:ISEvsSSE}.
The defender has only one resource. One can verify, the SSE strategy uniformly allocates this resource on $s_1,...,s_{n-1}$ in order to make $t_2$ be in the attack set. Unfortunately, $t_2$ is not inducible since $t_2$ is weakly dominated by $t_1$ for the attacker.
Therefore, $U^I_d(\mathbf{x})=\frac{1}{n-1}$. On the other hand, the ISE strategy uniformly assigns the resource on $s_1$ and $s_n$ which together cover all the targets and $U^I_d(\mathbf{x}')=\frac{1}{2}$.

\noindent
\textbf{Example 2.}
Notice that, in previous examples shown in the paper, a lot of targets have equal payoffs. However, this is only for the convenience of exposition. It is possible that ISE is not SSE when all targets have unequal payoffs. For example, there are 4 targets $t_1$, $t_2$, $t_3$, and $t_4$. Payoffs for the attacker on a successful attack are 1, 2, 3, and 4, respectively, and on an unsuccessful attack are -1, -2, -3, and -8 respectively. Payoffs for the defender on preventing an attack are 1, 100, 2, and 30 respectively and for failing to cover the attacked target are -1, 0, -2, and -3 respectively. There are two schedules $s_1=\{t_1,t_2,t_3\}$ and $s_2=\{t_4\}$, and one resource available to the defender. We can easily verify that SSE strategy is $\mathbf{x}^{SSE}=\langle0.5,0.5\rangle$ and attacker is assumed to attack $t_2$. However $t_2$ is not inducible, and ISE strategy is $\mathbf{x}^{ISE}=\langle9/14,5/14\rangle$ and attacker is induced to attack $t_4$.
\section{Computing an ISE}
We have shown that ISE mitigates the inducibility risk of SSE which can cause extremely worse performance in utility guarantee.
In this section, we further our discussion and show that, from a computational perspective, ISE does not complicate existing solution concepts as the problem of computing an ISE polynomially reduces to that of computing an SSE on the same class of schedules.
In addition to the theoretical result, a practical approach is also presented to compute an ISE.

\subsection{\mbox{A Polynomial-time Reduction to Computing an SSE}}

We start by defining the \emph{feasibility} of a target.
We say a target is \emph{feasible} if there exists $\mathbf{x} \in \mathcal{X}$ such that $t\in \Gamma(\mathbf{x})$.
We will henceforth refer to the problem of deciding if a target is feasible or not, the \emph{feasibility problem}; and of deciding if a target is inducible or not, the \emph{inducibility problem}.
The feasibility and inducibility of an element follow the similar definitions. We first restrict the investigation to games without identical targets. The reduction is presented in Theorem~\ref{thm:ISEtoSSE}, where a series of feasibility checks are incorporated as sub-procedures.

\begin{Lemma}\label{lmm:idc2fsb}
For any target $t$ in security games, the inducibility problem reduces to the feasibility problem on games with the same class of schedules in polynomial time.
\end{Lemma}

\begin{proof}[Proof sketch]
The intuition behind Lemma~\ref{lmm:idc2fsb} is the observation that whenever $U_a(\mathbf{c},t) > U_a(\mathbf{c},t')$ for all $t'\neq t$, there is a lower bound $\delta$ of the gap, such that $U_a(\mathbf{c},t) - U_a(\mathbf{c},t') \ge \delta$ for all $t'\neq t$, and $\log{\frac{1}{\delta}}$ is bounded by a polynomial in the input size. Blending $\delta$ into the payoffs, we construct a new game such that $t$ is inducible in original game if and only if $t$ is feasible in the new constructed game.
\end{proof}

\begin{Theorem}
\label{thm:ISEtoSSE}
    The problem of computing an ISE reduces to the problem of computing an SSE of games with the same class of schedules in polynomial time.
\end{Theorem}
\begin{proof}
    An ISE can be computed in the following way:
    \begin{itemize}
      \item[1.] Check inducibility of every targets and obtain $T^I$.
      \item[2.] For each $t \in T^I$, solve $\max_{\mathbf{x}: t \in \Gamma(\mathbf{x}) } U_d(\mathbf{x}, t)$, which yields the defender's optimal strategy under the constraint that $t$ is an optimal response of the attacker.

      \item[3.] Among all the solutions obtained above, find out the one with the highest defender utility. The corresponding target, $t^*$ say, and the optimal defender strategy corresponding to $t^*$ forms an ISE.
    \end{itemize}

    Specifically, in Step 1, the inducibility problem reduces to the feasibility problem by Lemma~\ref{lmm:idc2fsb}. The feasibility of $t$ can further be decided by computing the SSE of a game in which defender's payoff parameters are modified to: $U_d^u(t') = 1$ and $U_d^c(t') = 2$ for all $t' \neq t$; and $U_d^u(t) = 3$ and $U_d^c(t) = 4$ (the attacker's payoffs remain the same as in the feasibility problem). In this game, even the penalty on $t$ is strictly higher than the rewards on all the other targets, so the defender strictly prefers the attacker to choose $t$, irrespective of the coverage of the targets.
    Therefore, $t$ is feasible if $t$ is in the attacker's attack set in every SSE, so we can check whether this is true to decide the feasibility of $t$.

    In Step 2, each of the optimizations can be solved, again, by computing the SSE of a game in which defender's payoff parameters are modified to: $U_d^u(t') = 1$ and $U_d^c(t') = 2$ for all $t' \neq t$; and $U_d^u(t) = 3$ and $U_d^c(t) = 4$ (the attacker's payoffs remain the same as in the \emph{original game}).
    $t$ is inducible and hence feasible in the original game, and the feasibility remains in modified game as the attacker's payoffs are the same.
    For the same reason above, an SSE must incorporate $t$ in the attack set, so that $U_a(\mathbf{x}, t) \ge U_a(\mathbf{x}, t')\ \forall\,t'\neq t$ is satisfied. In addition, $c_t$ is maximized in the solution, so the SSE is exactly a solution to the optimization in Step 2.

    Therefore, an ISE is obtained via polynomially many calls to the computation of an SSE. This completes the proof.
\end{proof}

\subsection{Dealing with Identical Targets}

In the presence of identical targets, it is assumed that, for every inducible element, the target worst for the defender is to be chosen by the attacker.
We keep Step 1 of the procedure in the proof of Theorem~\ref{thm:ISEtoSSE} by treating identical targets as one target, so that a target is inducible if at least one target identical to it is inducible (even though this target might not actually be induced).
However, when we actually compute the defender's optimal strategy conditioned on a particular inducible target being attacked as in Step 2, we need additional constraints that require this target is worst for the defender among all targets that is identical to it, i.e.,
\[
    U_d(c_t, t) \le U_d(c_{t'}, t'), \quad \forall \ t' \text{ identical to } t.
\]
We convert the constraints to equivalent ones in the form of $U_a(\mathbf{x}, t) \ge U_a(\mathbf{x}, t')$ (as in Step 2) to finish the reduction.

Observe that $c_t = c_{t'}$ since $t'$ and $t$ are identical,
so the above constraints are equivalent to
\[
    U_d(c_t, t) \le U_d(c_{t}, t'), \quad \forall \ t' \text{ identical to } t,
\]
which only involve a single variable $c_t$.
Thus, the constraints effectively reduce to an inequality of the form $\alpha \le c_t \le \beta$, with two constants $\alpha$ and $\beta$.
For $c_t \ge \alpha$, given the objective of the problem as maximizing $U_d(\mathbf{x}, t)$ (which increases with $c_t$), this part can be ignored: if the solution does not satisfy $c_t \ge \alpha$, that means there is no feasible solution satisfying $c_t \ge \alpha$; we simply skip $t$ in Step 3.
The second half, $c_t \le \beta$, can be captured by modifying the attacker's payoff parameters of an arbitrary target $t'$, that is identical to $t$, to ${U}_a^c(t') = {U}_a^u(t') = U_a(\beta, t)$, so that the constraint $U_a(c_t, t) \ge U_a(c_{t'}, t')$ is now equivalent to $c_t \le \beta$ (this constraint is useless before the modification when it always holds that $U_a(c_t, t) = U_a(c_{t'}, t')$).

\subsection{Algorithmic Implementation}

As a theoretical result, the above reduction involves repeated calls of computing an SSE and therefore falls short on practical performance.
We introduce a more concise practical approach to compute an ISE. We first limit our scope to the games without identical target, for the ease of reading. The extension to include identical targets into consideration is fairly straightforward.

First, inducibility of a target $t$ can be decided using the following program: $t$ is inducible iff the optimum $u^* > 0$.
\begin{equation}\label{eq:inducible_t}
\begin{aligned}
\max_{\mathbf{x},u}\quad&u \\ 
\text{s.t.}\quad& U_a(\mathbf{x},t)\geq U_a(\mathbf{x},t')+u\quad\forall t'\neq t \\ 
&\sum\nolimits_{j\in J}x_j=1 
\end{aligned}
\end{equation}
Solving the above program for each $t \in T$, we obtain the inducible target set $T^I$.
By Proposition~\ref{prp:chaiseOptimal}, the computation of an ISE further converts to computing an SSE of a game restricted to targets in $T^I$.
There is a large body of research on designing algorithms for computing an SSE of security games with various types of schedules, such as ASPEN~\cite{Jain10} and CLASPE~\cite{gan15}.
These algorithms can be applied directly.
We note that the above approach requires altering schedules in the original game, while all schedules remain the same throughout our theoretical reduction.


Implied by Proposition~10 to compute the ``pessimistic" leader's payoff~\cite{Stengel10}, we can directly compute an SSE in a restricted game $g'$ whose target set is the set of inducible targets in targeted game $g$, and map this SSE to ISE of game $g$.
\begin{Proposition}\label{prp:chaiseOptimal}
For a security game $g = \langle T,R,S \rangle$, an SSE defender strategy of the game $g' = \langle T^I,R,S^I \rangle$ is an ISE strategy in $g$, where $T^I$ is the inducible target set of $g$, and $S^I=\{ s\cap T^I|s\in S\}$.
\end{Proposition}

When identical targets exist, we first enumerate all inducible elements $\mathcal{E}^I$ by solving optimization~\eqref{eq:inducible_t} with slight modification, by replacing the target $t$ and utility function $U_{\dag}(\mathbf{x},t)$ with the element $e$ and element-based utility function $\Tilde{U}_{\dag}(\mathbf{x},e)$ defined in~\eqref{eq:element_utility}, for $\dag\in\{a,d\}$. An ISE can be computed with the multi-LP approach~\cite{Conitzer06}, where each LP corresponds to an inducible element $e\in\mathcal{E}^I$ as follows
\begin{equation}\label{eq:ise_identical}
\begin{aligned}
\max_{\mathbf{x},u}\quad&u \\ 
\text{s.t.}\quad& \tilde{U}_a(\mathbf{x},e)\geq \tilde{U}_a(\mathbf{x},e')\quad&\forall e'\in\mathcal{E} \\ 
& u\leq U_d(\mathbf{x},t)\quad&\forall t\in e\\
&\sum\nolimits_{j\in J}x_j=1 
\end{aligned}
\end{equation}
The solution with the highest objective among multiple LPs is an ISE. It can be easily verified that the large body of designing algorithms, especially those based on strategy generation techniques, can adapt to solve~\eqref{eq:ise_identical} with little effort.

\section{Experimental Evaluation}
We evaluate our solution concept and proposed algorithmic implementation with extensive experiments. All results are obtained on a platform with a
2.60 GHz dual-core CPU and 8.0 GB memory. All linear programs are solved using the existing solver CPLEX (version 12.4). The random instances are generated as follows: rewards and penalties are all integers randomly drawn from $[0,5]$ and $[-5,0]$ respectively. Each schedule is randomly generated covering a fixed number $l$ of targets and each target is ensured to be covered by at least one schedule. The resources are all homogeneous, i.e., $S_r=S$ for any $r\in R$. Unless otherwise specified, all results are averaged on 100 randomly generated instances.

For the purpose of comparison, we define the \emph{overoptimism} and \emph{sub-optimality} of SSE w.r.t. the utility guarantee.
\begin{Definition}[\textbf{Overoptimism and sub-optimality}]
Let $\mathbf{x}$ be an SSE strategy.
\begin{itemize}
  \item $\mathbf{x}$ is overoptimistic if $U_d(\mathbf{x},f^S(\mathbf{x}))>U_d^I(\mathbf{x})$;
  \item $\mathbf{x}$ is sub-optimal if $U_d^I(\mathbf{x})<\max_{\mathbf{x}'\in\mathcal{X}}U_d^I(\mathbf{x}')$.
\end{itemize}
\end{Definition}

\begin{figure}[h]
\center
\subfigure{
    \label{fig:experiments:inducibility} 
    \includegraphics[width=\figwidth]{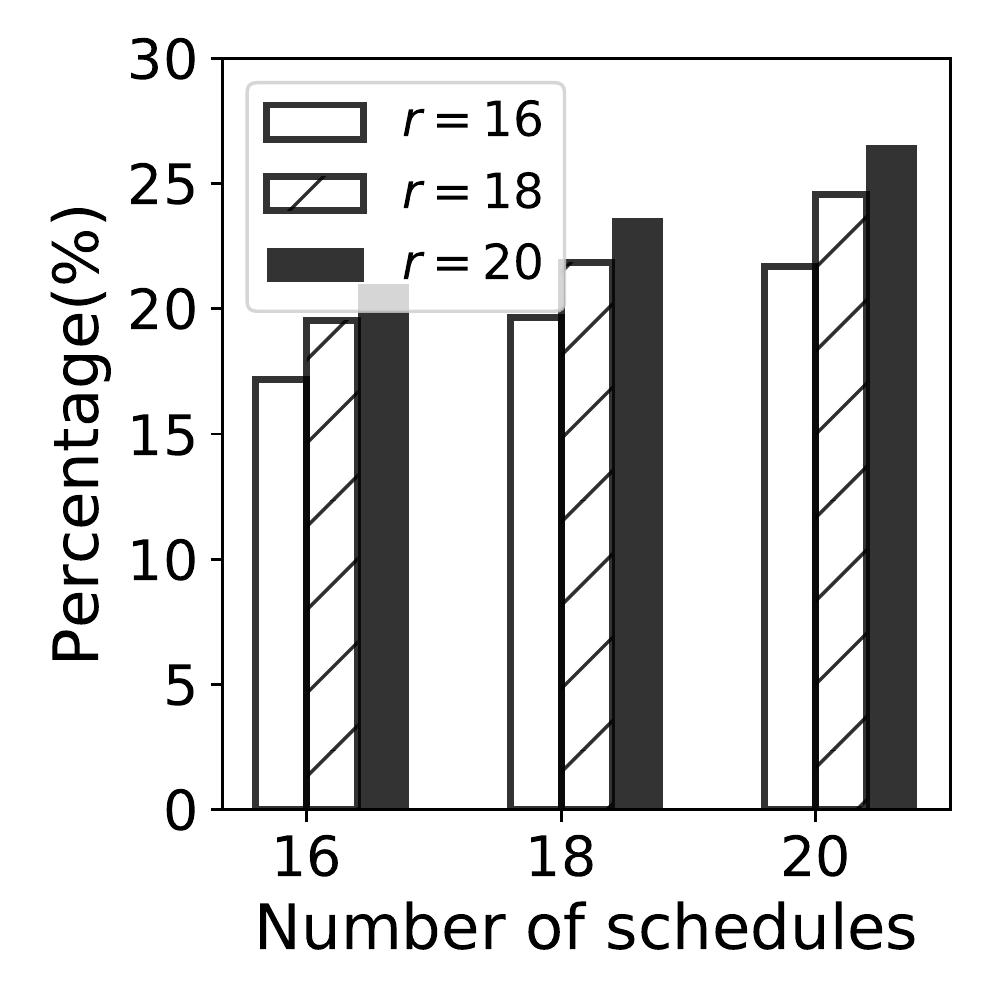}}
    \subfigure{
    \label{fig:experiments:scalability} 
    \includegraphics[width=\figwidth]{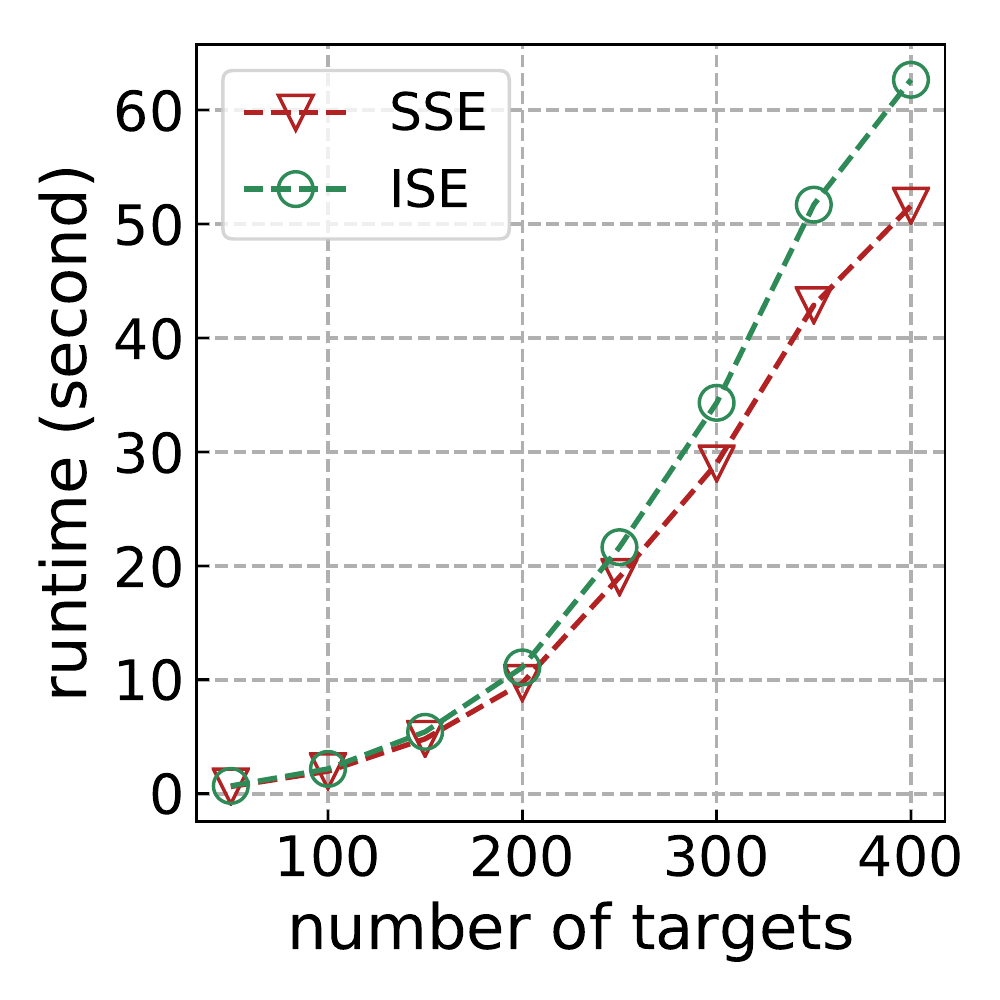}}
    \caption{\small{Inducibility (left) and scalability (right).}}\label{fig:experiments:indandsca}
\end{figure}
\noindent

\textbf{Inducibility}\quad
We depict the percentage of inducible targets on instances with 100 targets and 1 resource on the left of Figure~\ref{fig:experiments:indandsca}. The results show that with more schedules and more targets per schedule, the game has more inducible targets. That is because the defender can cover the high valued targets with enough resources so that the low valued targets can be induced to become unique best responses. The important observation here is that the percentage is neither too high nor too low (within $[15\%,30\%]$), which indicates that the inducibility is not a trivial property.

\noindent
\textbf{Scalability}\quad
We evaluate the scalability of our algorithmic implementation for computing an ISE. The result is shown on the right of Figure~\ref{fig:experiments:indandsca}. The game instances are randomly generated with $l=5$, $|R|=5$, $|T|$ ranges from 50 to 400 with step size of 50, and $|S|=\frac{|T|}{2}$. We adopt the column generation approach with heuristic bounds for pruning to solve the large scale LPs~\cite{gan15}. As a comparison, the scalability of computing SSE with the same algorithmic framework is also depicted. The result shows that, it takes almost the same computational costs to compute an ISE and an SSE. The algorithmic implementation can compute ISE for large-scale instances. Thus, ISE successfully mitigates the inducibility issue of SSE without sacrificing the benefit of scalable algorithms for computng SSE.

\begin{figure}[th]
  \centering
  \subfigure{
    \label{fig:experiments:overestimatepercentage} 
    \includegraphics[width=\figwidth]{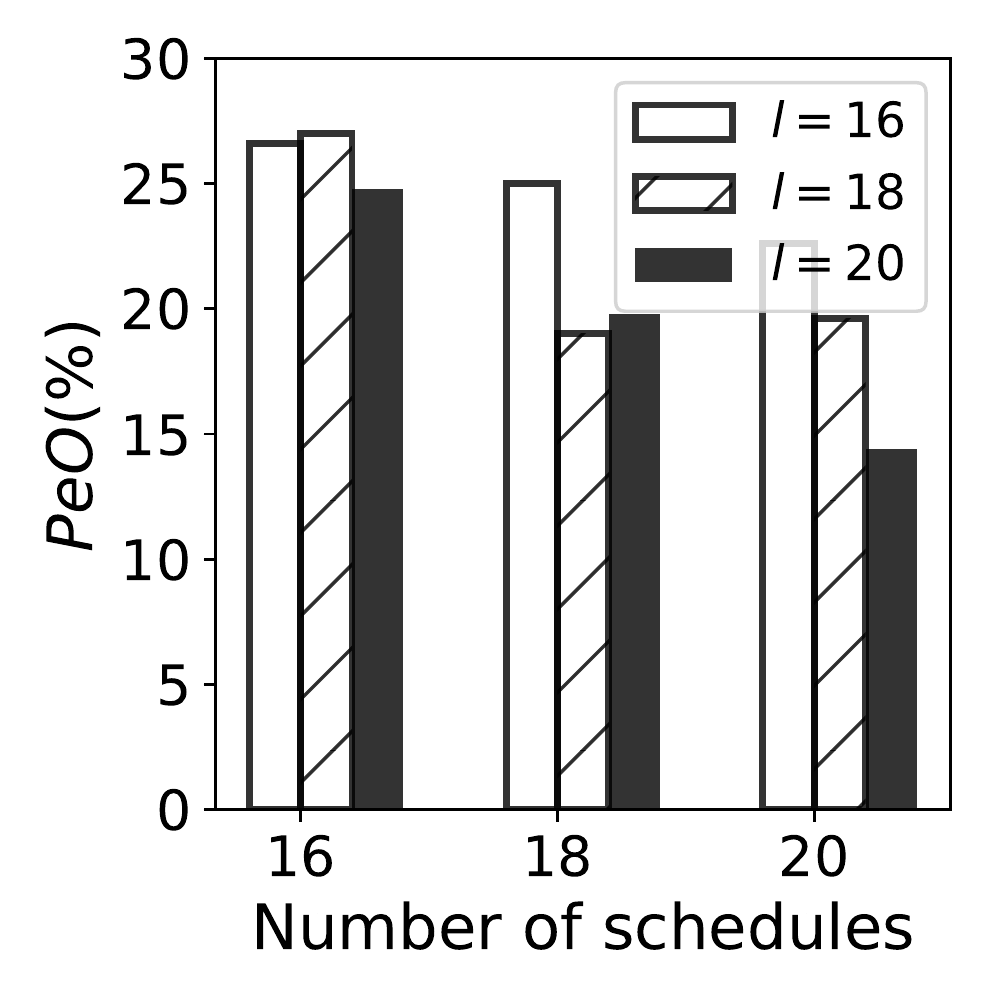}}
    \subfigure{
    \label{fig:experiments:suboptimalpercentage} 
    \includegraphics[width=\figwidth]{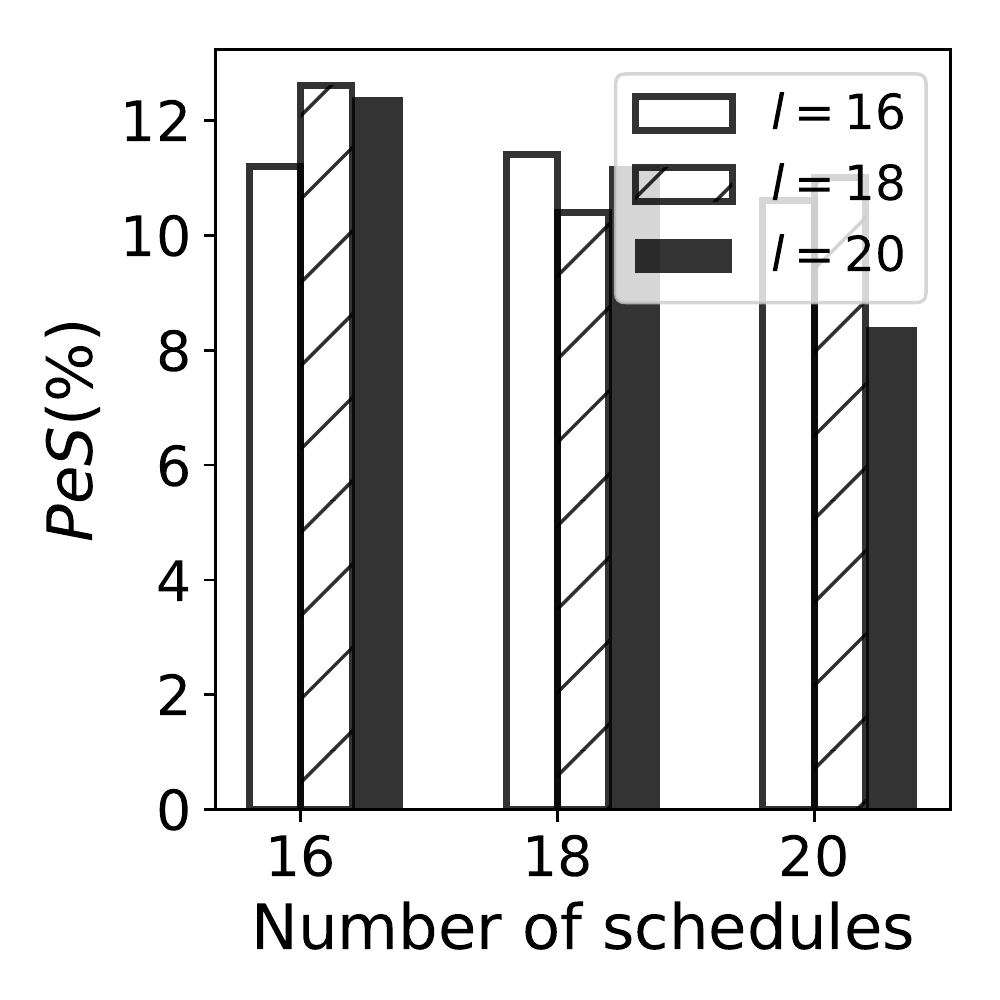}}
    \subfigure{
    \label{fig:experiments:overestimate} 
    \includegraphics[width=\figwidth]{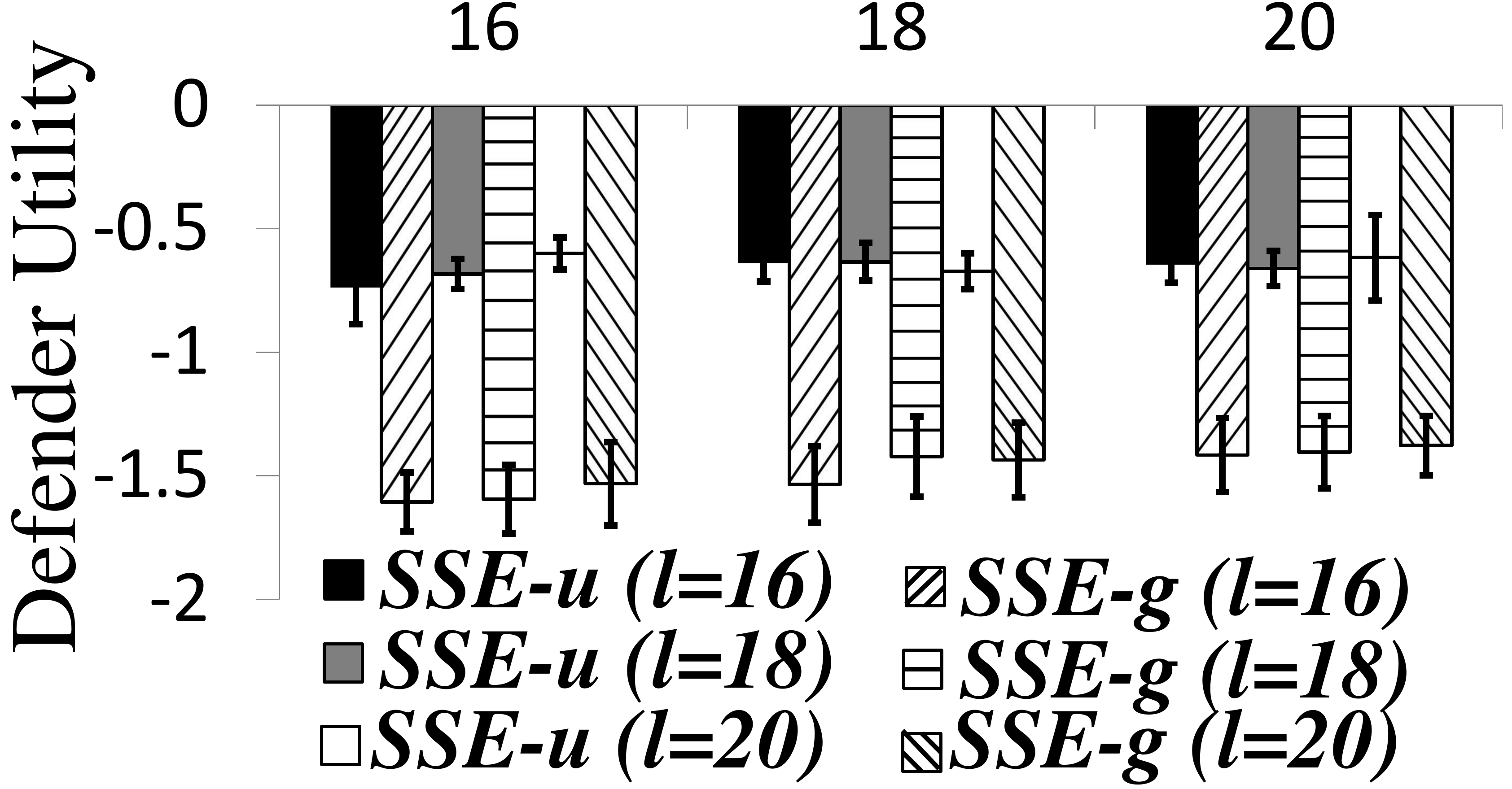}}
    \subfigure{
    \label{fig:experiments:suboptimal} 
    \includegraphics[width=\figwidth]{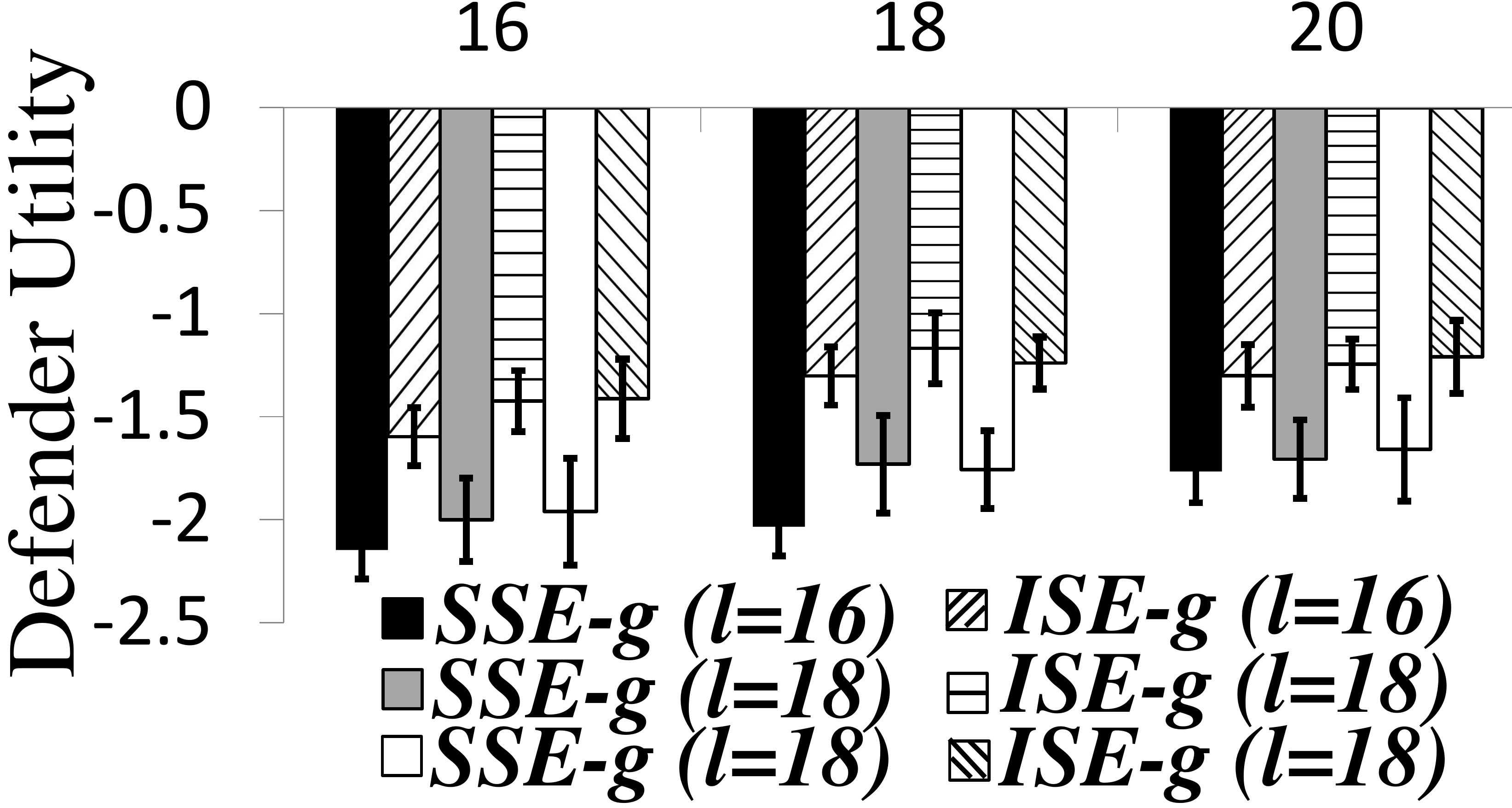}}
    \caption{\small{Overoptimism (left) and sub-optimality (right) of SSE.}}\label{fig:experiments:oversub}
\end{figure}
\noindent
\textbf{Overoptimism and Sub-optimality of SSE}\quad
We examine the overoptimism and sub-optimality of SSE. 500 instances are randomly generated with 200 targets, 1 resource, $|S|\in\{16,18,20\}$ and $l\in\{16,18,20\}$. This setting fits many realistic security domains, such as the port protection~\cite{Shieh12}, where the Coast Guard has few resources (patrol boats) and limited schedules due to complex geographic and efficiency constraints, and each schedule corresponds with one patrolling path visiting several targets. The results are shown in Figure~\ref{fig:experiments:oversub}, where \emph{PeO} and \emph{PeS} denote the percentages of instances with overoptimistic and sub-optimal SSE respectively. Moreover, Figure~\ref{fig:experiments:oversub} also shows the comparisons between the expected utility of SSE (``SSE-u") with the utility guarantee of SSE (``SSE-g") averaged on instances where SSE is overoptimistic, and similarly the comparisons between average utility guarantees of SSE and ISE (tagged with ``SSE-g" and ``ISE-g" respectively) on instances, where SSE is suboptimal. The 95\% confidence interval is depicted. The results show that SSE suffers from significant overoptimism and sub-optimality, which is highly problematic as we explained in the introduction.
We also conduct simulations in a large number of different parameter settings with 3 and more resources. Here we list the results on ten settings in the table on the right. 
\begin{wrapfigure}{r}{40mm}
  \begin{center}
    \includegraphics[width = 40mm]{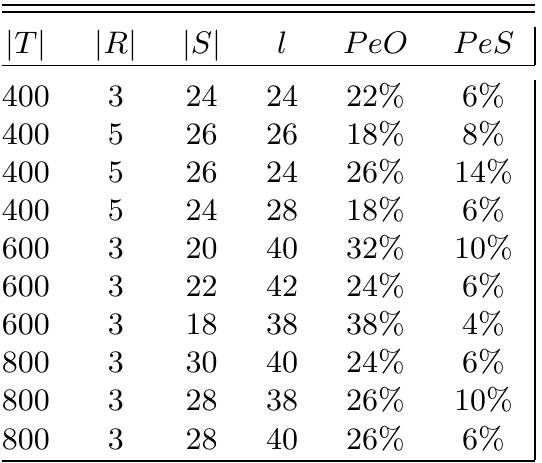}
  \end{center}
\end{wrapfigure}
For each of these settings, we randomly generate 50 instances.
Significant numbers of cases with overoptimistic and suboptimal SSE are observed for almost every setting.
Thus, the aforementioned risk of applying SSE in practice can be a general issue for many security domains and applications, and we argue that ISE should be considered as a ``safer'' alternative.
\section{Conclusion}
This paper reveals the significant potential risk of overoptimism of SSE in security games. We propose a new solution concept, ISE, by exploiting the inducible targets. Our theoretical analysis proves the existence of ISE and its optimality in utility guarantee, and our formal comparisons between ISE and SSE emphasize that ISE is a more suitable solution concept in security games. Extensive evaluation shows that SSE is significantly overoptimistic and ISE achieves significantly higher utility guarantee than SSE. We will investigate the inducibility issues in generic games and Bayesian games in future work.

\section*{Acknowledgments}
This research was supported by MURI Grant W911NF-11- 1-0332 and the National Research Foundation, Prime Minister's Office, Singapore under its IDM Futures Funding Initiative.
Jiarui Gan is supported by the EPSRC International Doctoral Scholars Grant EP/N509711/1. Tran-Thanh Long was supported by the EPSRC funded project EP/N02026X/.
\bibliographystyle{aaai}
\bibliography{reference}

\begin{thebibliography}{}

\bibitem[\protect\citeauthoryear{An}{2017}]{An17}
An, B.
\newblock 2017.
\newblock Game theoretic analysis of security and sustainability.
\newblock In {\em IJCAI},  5111--5115.

\bibitem[\protect\citeauthoryear{Basilico \bgroup et al\mbox.\egroup
  }{2017}]{Basilico17}
Basilico, N.; Celli, A.; Nittis, G.~D.; and Gatti, N.
\newblock 2017.
\newblock Coordinating multiple defensive resources in patrolling games with
  alarm systems.
\newblock In {\em AAMAS},  678--686.

\bibitem[\protect\citeauthoryear{Conitzer and Sandholm}{2006}]{Conitzer06}
Conitzer, V., and Sandholm, T.
\newblock 2006.
\newblock Computing the optimal strategy to commit to.
\newblock In {\em EC},  82--90.

\bibitem[\protect\citeauthoryear{Fang \bgroup et al\mbox.\egroup
  }{2016}]{Fang16}
Fang, F.; Nguyen, T.~H.; Pickles, R.; Lam, W.~Y.; Clements, G.~R.; An, B.;
  Singh, A.; Tambe, M.; and Lemieux, A.
\newblock 2016.
\newblock Deploying {PAWS:} {F}ield optimization of the protection assistant
  for wildlife security.
\newblock In {\em IAAI},  3966--3973.

\bibitem[\protect\citeauthoryear{Gan, An, and Vorobeychik}{2015}]{gan15}
Gan, J.; An, B.; and Vorobeychik, Y.
\newblock 2015.
\newblock Security games with protection externalities.
\newblock In {\em AAAI},  914--920.

\bibitem[\protect\citeauthoryear{Gan \bgroup et al\mbox.\egroup
  }{2017}]{gan17security}
Gan, J.; An, B.; Vorobeychik, Y.; and Gauch, B.
\newblock 2017.
\newblock Security games on a plane.
\newblock In {\em AAAI},  530--536.

\bibitem[\protect\citeauthoryear{Hohn}{2013}]{hohn2013elementary}
Hohn, F.~E.
\newblock 2013.
\newblock {\em Elementary matrix algebra}.
\newblock Courier Corporation.

\bibitem[\protect\citeauthoryear{Jain \bgroup et al\mbox.\egroup
  }{2010}]{Jain10}
Jain, M.; Kardes, E.; Kiekintveld, C.; Ord{\'{o}}{\~{n}}ez, F.; and Tambe, M.
\newblock 2010.
\newblock Security games with arbitrary schedules: {A} branch and price
  approach.
\newblock In {\em AAAI},  792--797.

\bibitem[\protect\citeauthoryear{Kiekintveld \bgroup et al\mbox.\egroup
  }{2009}]{Kiekintveld09}
Kiekintveld, C.; Jain, M.; Tsai, J.; Pita, J.; Ord{\'{o}}{\~{n}}ez, F.; and
  Tambe, M.
\newblock 2009.
\newblock Computing optimal randomized resource allocations for massive
  security games.
\newblock In {\em AAMAS},  689--696.

\bibitem[\protect\citeauthoryear{Korzhyk \bgroup et al\mbox.\egroup
  }{2011}]{Korzhyk11}
Korzhyk, D.; Yin, Z.; Kiekintveld, C.; Conitzer, V.; and Tambe, M.
\newblock 2011.
\newblock Stackelberg vs. {N}ash in security games: {A}n extended investigation
  of interchangeability, equivalence, and uniqueness.
\newblock {\em J. Artif. Intell. Res.} 41:297--327.

\bibitem[\protect\citeauthoryear{Leitmann}{1978}]{leitmann78}
Leitmann, G.
\newblock 1978.
\newblock On generalized {S}tackelberg strategies.
\newblock {\em Journal of Optimization Theory and Applications} 26(4):637--643.

\bibitem[\protect\citeauthoryear{McCarthy \bgroup et al\mbox.\egroup
  }{2016}]{McCarthy16}
McCarthy, S.~M.; Tambe, M.; Kiekintveld, C.; Gore, M.~L.; and Killion, A.
\newblock 2016.
\newblock Preventing illegal logging: Simultaneous optimization of resource
  teams and tactics for security.
\newblock In {\em AAAI},  3880--3886.

\bibitem[\protect\citeauthoryear{Nguyen \bgroup et al\mbox.\egroup
  }{2013}]{nguyen13}
Nguyen, T.~H.; Yang, R.; Azaria, A.; Kraus, S.; and Tambe, M.
\newblock 2013.
\newblock Analyzing the effectiveness of adversary modeling in security games.
\newblock In {\em AAAI},  718--724.

\bibitem[\protect\citeauthoryear{Okamoto, Hazon, and Sycara}{2012}]{Okamoto12}
Okamoto, S.; Hazon, N.; and Sycara, K.~P.
\newblock 2012.
\newblock Solving non-zero sum multiagent network flow security games with
  attack costs.
\newblock In {\em AAMAS},  879--888.

\bibitem[\protect\citeauthoryear{Paruchuri \bgroup et al\mbox.\egroup
  }{2008}]{Paruchuri08}
Paruchuri, P.; Pearce, J.~P.; Marecki, J.; Tambe, M.; Ord{\'{o}}{\~{n}}ez, F.;
  and Kraus, S.
\newblock 2008.
\newblock Playing games for security: {A}n efficient exact algorithm for
  solving {B}ayesian {S}tackelberg games.
\newblock In {\em AAMAS},  895--902.

\bibitem[\protect\citeauthoryear{Pita \bgroup et al\mbox.\egroup
  }{2008}]{Pita08}
Pita, J.; Jain, M.; Marecki, J.; Ord{\'{o}}{\~{n}}ez, F.; Portway, C.; Tambe,
  M.; Western, C.; Paruchuri, P.; and Kraus, S.
\newblock 2008.
\newblock Deployed {ARMOR} protection: the application of a game theoretic
  model for security at the {L}os {A}ngeles international airport.
\newblock In {\em AAMAS},  125--132.

\bibitem[\protect\citeauthoryear{Pita \bgroup et al\mbox.\egroup
  }{2009}]{Pita09}
Pita, J.; Jain, M.; Ord{\'{o}}{\~{n}}ez, F.; Tambe, M.; Kraus, S.; and
  Magori{-}Cohen, R.
\newblock 2009.
\newblock Effective solutions for real-world stackelberg games: when agents
  must deal with human uncertainties.
\newblock In {\em AAMAS},  369--376.

\bibitem[\protect\citeauthoryear{Sandholm}{2015}]{sandholm15}
Sandholm, T.
\newblock 2015.
\newblock Solving imperfect-information games.
\newblock {\em Science} 347(6218):122--123.

\bibitem[\protect\citeauthoryear{Shieh \bgroup et al\mbox.\egroup
  }{2012}]{Shieh12}
Shieh, E.; An, B.; Yang, R.; Tambe, M.; Baldwin, C.; DiRenzo, J.; Maule, B.;
  and Meyer, G.
\newblock 2012.
\newblock {PROTECT:} {A} deployed game theoretic system to protect the ports of
  the {U}nited {S}tates.
\newblock In {\em AAMAS},  13--20.

\bibitem[\protect\citeauthoryear{Tambe}{2011}]{tambe11}
Tambe, M.
\newblock 2011.
\newblock {\em Security and Game Theory - Algorithms, Deployed Systems, Lessons
  Learned}.
\newblock Cambridge University Press.

\bibitem[\protect\citeauthoryear{Tsai \bgroup et al\mbox.\egroup
  }{2009}]{tsai09}
Tsai, J.; Kiekintveld, C.; Ordonez, F.; Tambe, M.; and Rathi, S.
\newblock 2009.
\newblock {IRIS}-{A} tool for strategic security allocation in transportation
  networks.
\newblock In {\em AAMAS},  37--44.

\bibitem[\protect\citeauthoryear{Varakantham, Lau, and
  Yuan}{2013}]{Varakantham13}
Varakantham, P.; Lau, H.~C.; and Yuan, Z.
\newblock 2013.
\newblock Scalable randomized patrolling for securing rapid transit networks.
\newblock In {\em IAAI},  1563--1568.

\bibitem[\protect\citeauthoryear{Von~Stengel and Zamir}{2004}]{von04}
Von~Stengel, B., and Zamir, S.
\newblock 2004.
\newblock Leadership with commitment to mixed strategies.
\newblock {\em Technical Report LSE-CDAM-2004-01, CDAM Research Report.}

\bibitem[\protect\citeauthoryear{von Stengel and Zamir}{2010}]{Stengel10}
von Stengel, B., and Zamir, S.
\newblock 2010.
\newblock Leadership games with convex strategy sets.
\newblock {\em Games and Economic Behavior} 69(2):446--457.

\bibitem[\protect\citeauthoryear{Xu \bgroup et al\mbox.\egroup }{2017}]{Xu17}
Xu, H.; Ford, B.~J.; Fang, F.; Dilkina, B.; Plumptre, A.~J.; Tambe, M.;
  Driciru, M.; Wanyama, F.; Rwetsiba, A.; Nsubaga, M.; and Mabonga, J.
\newblock 2017.
\newblock Optimal patrol planning for green security games with black-box
  attackers.
\newblock In {\em GameSec},  458--477.

\bibitem[\protect\citeauthoryear{Yang \bgroup et al\mbox.\egroup
  }{2014}]{Yang14}
Yang, R.; Ford, B.~J.; Tambe, M.; and Lemieux, A.
\newblock 2014.
\newblock Adaptive resource allocation for wildlife protection against illegal
  poachers.
\newblock In {\em AAMAS},  453--460.

\bibitem[\protect\citeauthoryear{Yin, An, and Jain}{2014}]{Yin14}
Yin, Y.; An, B.; and Jain, M.
\newblock 2014.
\newblock Game-theoretic resource allocation for protecting large public
  events.
\newblock In {\em AAAI},  826--834.

\end{thebibliography}

\appendix

\section{Appendix}

\subsection{Proof of Lemma~\ref{lmm:idc2fsb}}

W.l.o.g., we assume all payoff parameters are integers encoded in binary. To prove {Lemma~\ref{lmm:idc2fsb}}, we need the following preliminary results.

\smallskip
\noindent
\textbf{Claim 1.} \emph{Suppose $A$ is an $n \times n$ invertible matrix and $\mathbf{b}$ is a vector of size $n$; all the entries of $A$ and $\mathbf{b}$ are integers and have their absolute values bounded by $M$. 
Then every component of $A^{-1} \mathbf{b}$ can be written as a fraction $p/q$ with $p$ and $q$ being integers, $|p| \le (nM)^n$ and  $|q| \le (nM)^n$.}
\smallskip

\begin{proof}
Let $\mathbf{x} = A^{-1} \mathbf{b}$. From Cramer's rule, we have
\[
    x_i = \frac{\det(A^{(i)})}{\det(A)},
\]
where $A^{(i)}$ is $A$ with the $i$-th column replaced by $b$.
Expanding the determinants, we have
\[
    \det(A) = \sum_{\sigma} (-1)^{|\sigma|} \prod_{i=1}^n a_{i,\sigma(i)},
\]
where the summation is over all $n!$ permutations of $(1,\dots,n)$, and $|\sigma|$ denotes the number of inversions of the permutation $\sigma$.
It follows that
\[
    |\det(A)| \le \sum_{\sigma} \prod_{i=1}^n \le n! M^n \le (nM)^n.
\]
The same applies to $|\det(A^{(i)})|$.
\end{proof}

\smallskip

\newcommand{\Mbound}{2(n+1) ( n^2 M_0 )^{n^2}}

\noindent
\textbf{Claim 2.} \emph{
Let $\mathcal{C}' = \{ \mathbf{c} : \mathbf{c} \in\mathcal{C} \wedge U_a(\mathbf{c},t) \ge U_a(\mathbf{c},t')\ \forall t'\neq t\}$.
Every vertex of $\mathcal{C}'$, as a vector, can be written in the form $(\frac{p_1}{q_1},\dots,\frac{p_n}{q_n})$ with each $p_t$ and $q_t$ being integers bounded by $\Mbound$, where $M_0$ is the bound of the payoff parameters.}
\smallskip

\begin{proof}
Let $\mathcal{C}^U := \{\mathbf{c} : U_a(\mathbf{c},t) \ge U_a(\mathbf{c},t')\ \forall t'\neq t\}$, so equivalently, $\mathcal{C}' = \mathcal{C} \cap \mathcal{C}^U$.
Consider a vertex $\mathbf{v}$ of $\mathcal{C}'$. 
Since $\mathbf{v} \in \mathcal{C}$, $\mathbf{v}$ must be supported on $\tilde{n} +1$ \emph{vertices} of $\mathcal{C}$, say $\mathbf{v}_0, \mathbf{v}_1, \dots, \mathbf{v}_{\tilde{n}}$ and can be written as the convex combination: $\mathbf{v} = \sum_{i = 0}^{\tilde{n}} \alpha_i \mathbf{v}_i$, where $\alpha_i > 0$ and $\sum_{i = 0}^{\tilde{n}} \alpha_i =1$;
In particular, we pick the $\tilde{n}$ such that no support set of size smaller than $\tilde{n} +1$ exists, so that $\mathbf{v}_0, \dots, \mathbf{v}_{\tilde{n}}$ are affine independent, or in other words, $\mathbf{v}_1 - \mathbf{v}_0, \mathbf{v}_2 - \mathbf{v}_0, \dots, \mathbf{v}_{\tilde{n}}- \mathbf{v}_0$ are linear independent.

If $\tilde{n} = 0$, then $\mathbf{v}$ is simply a vertex of $\mathcal{C}$ and is in $\{0,1\}^n$ as $\mathcal{C}$ is the convex hull of the pure strategy set $\mathbf{P} \subseteq \{0,1\}^n$.
Obviously, $\mathbf{v}$ is in the fractional form $(\frac{p_1}{q_1},\dots,\frac{p_n}{q_n})$ with $p_t$ and $q_t$ bounded by $\Mbound$.

It remains to consider the case when $\tilde{n} > 0$.
Since $\mathbf{v} \in \mathcal{C}^U$, it holds that $U_a(\mathbf{v},t) \ge U_a(\mathbf{v},t')$ for all $t' \neq t$.
Particularly, for these inequalities, we pick out those \emph{not strictly} satisfies by $\mathbf{v}$ and arrange them in the form ${A} \mathbf{c} \ge \mathbf{b}$ (so that ${A} \mathbf{v} = \mathbf{b}$); similarly, those strictly satisfies in the form ${B} \mathbf{c} \ge \mathbf{d}$ (so that ${B} \mathbf{v} > \mathbf{d}$).
We have $\operatorname{rank}(A) \ge \tilde{n}$ since otherwise $\operatorname{rank}(AV) \le \operatorname{rank}(A) < \min\{ n_{\textnormal{row}}(A), \tilde{n}\}$ for $V = \begin{bmatrix} \mathbf{v}_1 - \mathbf{v}_0 & \mathbf{v}_2 - \mathbf{v}_0 & \cdots & \mathbf{v}_{\tilde{n}} - \mathbf{v}_0 \\ \end{bmatrix}$, where $n_{\textnormal{row}}(A)$ is the number of rows of $A$; so that $AV$ does not have full rank and there will exist infinitely many $\vec{\alpha'}$ satisfying
\begin{equation}
\label{eq:AValp}
    A V \vec{\alpha}' = \mathbf{b} - A \mathbf{v}_0.
\end{equation}
We show that this will lead to a contradiction.
Denote $\vec{\alpha} = {\begin{bmatrix}\alpha_{1}\; \alpha_{2}\; \cdots \; \alpha_{\tilde{n}}\end{bmatrix}}^{\rm {T}}$ and pick one $\vec{\alpha}' \neq \vec{\alpha}$ that satisfies Eq.~\eqref{eq:AValp}.
Let $\vec{\alpha}^{(1)} = (1-\epsilon) \vec{\alpha} + \epsilon \vec{\alpha}'$ and $\vec{\alpha}^{(2)} = (1+\epsilon) \vec{\alpha} - \epsilon \vec{\alpha}'$; and let $\vec{\nu}^{(1)} = (1 - \sum_{i = 1}^{\tilde{n}} \alpha_i^{(1)}) \mathbf{v}_0 + \sum_{i = 1}^{\tilde{n}} \alpha_i^{(1)} \mathbf{v}_i$ and $\vec{\nu}^{(2)} = (1 - \sum_{i = 1}^{\tilde{n}} \alpha_i^{(2)}) \mathbf{v}_0 + \sum_{i = 1}^{\tilde{n}} \alpha_i^{(2)} \mathbf{v}_i$.
Observe the follows:
\begin{itemize}
\item Given that $\vec{\alpha} > {0}$ and $1 - \sum_{i = 1}^{\tilde{n}} \alpha_i = \alpha_0 > 0$, we can have $\vec{\alpha}^{(1)}$ and $\vec{\alpha}^{(2)}$ satisfy the same by choosing an $\epsilon$ sufficiently close to $0$, so that each one of $\vec{\alpha}^{(1)}$ and $\vec{\alpha}^{(2)}$ defines a convex combination, and both $\vec{\nu}^{(1)}$ and $\vec{\nu}^{(2)}$, as convex combinations of $\mathbf{v}_0,\dots, \mathbf{v}_{\tilde{n}}$, will be in $\mathcal{C}$.

\item Similarly, when $\epsilon$ is sufficiently close to $0$ we can have $\vec{\nu}^{(1)}$ and $\vec{\nu}^{(2)}$ arbitrarily close to $\mathbf{v}$, so that ${B} \vec{\nu}^{(1)} > \mathbf{d}$ and ${B} \vec{\nu}^{(2)} > \mathbf{d}$ will hold.

\item $\vec{\alpha}$ satisfy Eq.~\eqref{eq:AValp}, which is easy to verify, and $\vec{\alpha}$ satisfy Eq.~\eqref{eq:AValp} by definition, so that as linear combinations of $\vec{\alpha}$ and $\vec{\alpha}'$, $\vec{\alpha}^{(1)}$ and $\vec{\alpha}^{(2)}$ also satisfy Eq.~\eqref{eq:AValp}; this immediately gives $A\vec{\nu}^{(1)} = \mathbf{b}$ and $A\vec{\nu}^{(2)} = \mathbf{b}$, and in turn $\begin{bmatrix} A\\B \end{bmatrix} \vec{\nu}^{(1)} \ge \begin{bmatrix} \mathbf{b}\\\mathbf{d} \end{bmatrix}$ and $\begin{bmatrix} A\\B \end{bmatrix} \vec{\nu}^{(2)} \ge \begin{bmatrix} \mathbf{b}\\\mathbf{d} \end{bmatrix}$, so that both $\vec{\nu}^{(1)}$ and $\vec{\nu}^{(2)}$ are in $\mathcal{C}^U$.
\end{itemize}
Therefore, by choosing an $\epsilon$ sufficiently close to $0$, we can have $\vec{\nu}^{(1)} \in \mathcal{C} \cap \mathcal{C}^U = \mathcal{C}'$ and so is $\vec{\nu}^{(2)}$; however, since $\mathbf{v} = \frac{1}{2} \vec{\nu}^{(1)} + \frac{1}{2} \vec{\nu}^{(2)}$ and $\vec{\nu}^{(1)} \neq \vec{\nu}^{(2)} \neq \mathbf{v}$, this contradicts that $\mathbf{v}$ is a \emph{vertex} of $\mathcal{C}'$.
As a result, $\operatorname{rank}(A) \ge \tilde{n}$, and we can choose $\tilde{n}$ linear independent rows of $A$ to form a submatrix $A'$; denote also the corresponding rows of $\mathbf{b}$ be $\mathbf{b}'$.
We have $\operatorname{rank} (A') = \tilde{n}$ and $ A' V \vec{\alpha} = \mathbf{b}' - A' \mathbf{v}_0$.
Note that $\operatorname{rank} (V) = \tilde{n}$ as $\mathbf{v}_1 - \mathbf{v}_0, \mathbf{v}_2 - \mathbf{v}_0, \dots, \mathbf{v}_{\tilde{n}}- \mathbf{v}_0$ are linear independent.
By Sylvester's rank inequality~\cite{hohn2013elementary},
\[\operatorname {rank} (A')+\operatorname {rank} (V)- \tilde{n} \leq \operatorname {rank} (A'V),\]
so that $\operatorname {rank} (A'V) \ge \tilde{n}$ and
$A'V$ is a full-rank $\tilde{n} \times \tilde{n}$ matrix.
This gives $\vec{\alpha} = ( A' V )^{-1} ( \mathbf{b}' - A' \mathbf{v}_0)$.
By Claim 1, all $\alpha_i$ can be written in the form $p/q$ with $p$ and $q$ bounded by
\[
    M_1 = (n^2 M_0)^n,
\]
(note that the entries of $A' V$ and $\mathbf{b}' - A' \mathbf{v}_0$ are bounded by $nM_0$).
Moreover, given that $\mathbf{v} = \sum_{i=0}^{\tilde{n}} \alpha_i \mathbf{v}_i =  \mathbf{v}_0 + \sum_{i = 1}^{\tilde{n}} \alpha_i (\mathbf{v}_i - \mathbf{v}_0)$ and that $\mathbf{v}_i$ has components either $0$ or $1$, each component of $\mathbf{v}$ is in the form $p/q$ with $p$ and $q$ bounded by
\begin{align*}
    M_2 = 2(n+1) M_1^n = \Mbound. & \qedhere
\end{align*}
\end{proof}

\begin{proof}[Proof of Lemma~\ref{lmm:idc2fsb}]
To decide the inducibility of a target $t$ is to check if there exists some $\mathbf{c}\in\mathcal{C}$ such that $U_a(\mathbf{c},t) > U_a(\mathbf{c},t')$ for all $t'\neq t$.
%
We show that the existence of such a $\mathbf{c}$
leads to the existence of a $\mathbf{c}^* \in \mathcal{C}$ such that $U_a(\mathbf{c}^*,t) \ge U_a(\mathbf{c}^*,t') + \frac{1}{(n+1) M_2^2}$ for all $t' \neq t$, so that the problems transforms to one of verifying weak satisfaction.

Now suppose that $t$ is inducible, and $\mathbf{c} \in \mathcal{C}$ is such that $U_a(\mathbf{c},t) > U_a(\mathbf{c},t')$ for all $t'\neq t$.
Let $\mathcal{C}' = \{ \mathbf{s}: \mathbf{s} \in\mathcal{C} \wedge U_a(\mathbf{s},t) \ge U_a(\mathbf{s},t')\ \forall t'\neq t\}$.
We have $\mathbf{c} \in \mathcal{C}$ and, since $\mathcal{C}$ is closed, $\mathbf{c}$ is supported on $n+1$ {vertices} of $\mathcal{C}$, say $\mathbf{v}_1, \dots, \mathbf{v}_{n+1}$, in a way such that $\mathbf{c} = \sum_{i =1}^{n+1} \alpha_i \mathbf{v}_i$ for some nonnegative values $\alpha_i$ with $\sum_{i =1}^{n+1} \alpha_i =1$.
%
%
%
Obviously, $U_a(\mathbf{v}_i,t) \ge U_a(\mathbf{v}_i,t')$ for all $i$ since $\mathbf{v}_i \in \mathcal{C}'$.
Moreover, for each $t'$, at least one $\mathbf{v}_i$ must strictly satisfy the constraint as otherwise we would have $U_a(\mathbf{c},t) - U_a(\mathbf{c},t') = \sum_i \alpha_i (U_a(\mathbf{v}_i,t) - U_a(\mathbf{v}_i,t')) = 0$, contradicting that $U_a(\mathbf{c},t) > U_a(\mathbf{c},t')$.
By Claim~2, $\mathbf{v}_1, \dots, \mathbf{v}_{n+1}$ can be written in the form $(\frac{p_1}{q_1},\dots,\frac{p_n}{q_n})$ with each $p_t$ and $q_t$ bounded by $M_2 = \Mbound$.
It follows that whenever $U_a(\mathbf{c},t) - U_a(\mathbf{c},t') > 0$, we have $U_a(\mathbf{v}_i,t) - U_a(\mathbf{v}_i,t') > 0$ for some $\mathbf{v}_i$; or put it differently,
\[
    \textstyle \lambda_t \cdot \frac{p_t}{q_t} - \lambda_{t'} \cdot \frac{p_{t'}}{q_{t'}} + \gamma_{t,t'} > 0,
\]
where $\lambda_t = U_a^c(t) - U_a^u(t)$, $\lambda_{t'} = U_a^c(t') - U_a^u(t')$ and $\gamma_{t,t'} = U_a^u(t) - U_a^u(t')$.
Equivalently,
\[
    \lambda_t \cdot {p_t} {q_{t'}} - \lambda_{t'} \cdot {p_{t'}}{q_{t}} + \gamma_{t,t'} \cdot {q_{t}} {q_{t'}} > 0.
\]
Now that all the coefficients are integers, the above inequality further implies
\[
    \lambda_t \cdot {p_t} {q_{t'}} - \lambda_{t'} \cdot {p_{t'}}{q_{t}} + \gamma_{t,t'} \cdot {q_{t}} {q_{t'}} \ge 1,
\]
so that
\[
    \textstyle U_a(\mathbf{v}_i,t) - U_a(\mathbf{v}_i,t')  = \lambda_t \cdot \frac{p_t}{q_t} - \lambda_{t'} \cdot \frac{p_{t'}}{q_{t'}} + \gamma_{t,t'} \ge \frac{1}{q_t q_{t'}} \ge \frac{1}{M_2^2}.
\]

Now consider the point $\mathbf{c}^* = \sum_{i=1}^{n+1} \frac{1}{n+1} \mathbf{v}_i$.
We have $\mathbf{c}^* \in \mathcal{C}'$ as it is in the convex hull of $\mathbf{v}_1, \dots, \mathbf{v}_{n+1}$.
In addition, for all $t' \neq t$,
\begin{align*}
     & \textstyle U_a(\mathbf{c}^*,t) - U_a(\mathbf{c}^*,t') %
    = \sum_{i=1}^{n+1} \frac{1}{n+1} \left( U_a(\mathbf{v}_i,t) - U_a(\mathbf{v}_i,t') \right) \\%
     & \textstyle \quad \ge \max_{i \in [n+1]} \frac{1}{n+1} \left( U_a(\mathbf{v}_i,t) - U_a(\mathbf{v}_i,t') \right) \ge \frac{1}{(n+1)M_2^2},
\end{align*}
from which we obtain the following key component:
{\small
\begin{align}
    &\exists\, \mathbf{c} \in \mathcal{C}, \forall\, t' \neq t:  U_a(\mathbf{c},t) > U_a(\mathbf{c},t') \ \Longleftrightarrow \nonumber \\
    &\textstyle \exists\, \mathbf{c}^* \in \mathcal{C}, \forall\, t' \neq t:  U_a(\mathbf{c}^*,t) \ge U_a(\mathbf{c}^*,t') + \frac{1}{(n+1) M_2^2}. \label{eq:ctocstar}
\end{align}
}

To decide the inducibility of target $t$, that is, if there exists some $\mathbf{c}\in\mathcal{C}$ such that $U_a(\mathbf{c},t) > U_a(\mathbf{c},t')$ for all $t'\neq t$, we construct a new game parameterized by $\hat{U}_a^c(\cdot)$ and $\hat{U}_a^u(\cdot)$, such that:
\begin{itemize}
  \item $\hat{U}_a^c(t') = (n+1) M_2^2 \cdot U_a^c(t')$ and $\hat{U}_a^u(t') = (n+1) M_2^2 \cdot U_a^u(t')$ for all $t' \neq t$.

  \item $\hat{U}_a^c(t) = (n+1) M_2^2 \cdot U_a^c(t) - 1$ and $\hat{U}_a^u(t) = (n+1) M_2^2 \cdot U_a^u(t) - 1$.
\end{itemize}
In such a way, it holds that
\begin{align*}
    &\exists\, \mathbf{c} \in \mathcal{C}: \hat{U}_a(\mathbf{c},t) \ge \hat{U}_a(\mathbf{c},t') \Longleftrightarrow \\%
 &\exists\, \mathbf{c} \in \mathcal{C}: (n+1) M_2^2 \cdot U_a(\mathbf{c},t) \ge (n+1) M_2^2 \cdot U_a(\mathbf{c},t') + 1 \\& \overset{\text{by Eq~\eqref{eq:ctocstar}}}{\Longleftrightarrow} %
\exists\, \mathbf{c} \in \mathcal{C}: U_a(\mathbf{c},t) > U_a(\mathbf{c},t');
\end{align*}
Namely, the inducibility of target $t$ is equivalent to the feasibility of $t$ in the new game.
The length of the new parameters (as binaries) is bounded by $\log ( (n+1) M_2^2 \cdot M_0) = O( n^2 \log n + n^2 \log M_0)$, a polynomial in the original input size.
This is a polynomial-time reduction from the inducibility problem to the feasibility problem.
\end{proof}

\end{document}